\documentclass[a4paper,UKenglish,cleveref, autoref, thm-restate]{lipics-v2021}
\usepackage{url}
\usepackage{hyperref}
\usepackage{thm-restate}

\usepackage{microtype}
\usepackage[utf8]{inputenc}
\usepackage{amsmath, amsthm, amssymb, mathtools, stmaryrd}
\usepackage{tikz}
\usepackage{xspace}
\usepackage{todonotes}
\usepackage{enumerate}
\usepackage{listings}
\usepackage{algorithm}
\usepackage{wrapfig}
\usepackage{comment}
\usepackage{stmaryrd}

\usetikzlibrary{backgrounds}

\usepackage{algorithmic}

\makeatletter
\newcommand{\@abbrev}[3]{
	\def\c@a@def##1{
		\if ##1.
		\relax
		\else
		\@ifdefinable{\@nameuse{#1##1}}{\@namedef{#1##1}{#2##1}}
		\expandafter\c@a@def
		\fi
	}
	\c@a@def #3.
}
\@abbrev{bb}{\mathbb}{ABCDEFGHIJKLMNOPQRSTUVWXYZ}
\@abbrev{bf}{\mathbf}{ABCDEFGHIJKLMNOPQRSTUVWXYZabcdefghijklmnopqrstuvwxyz}
\@abbrev{bit}{\boldsymbol}{
	ABCDEFGHIJKLMNOPQRSTUVWXYZabcdefghijklmnopqrstuvwxyz}
\@abbrev{mc}{\mathcal}{ABCDEFGHIJKLMNOPQRSTUVWXYZ}
\@abbrev{mb}{\mathbb}{ABCDEFGHIJKLMNOPQRSTUVWXYZ}
\@abbrev{mf}{\mathfrak}{
	ABCDEFGHIJKLMNOPQRSTUVWXYZabcdefghijklmnopqrstuvwxyz}
\@abbrev{rm}{\mathrm}{ABCDEFGHIJKLMNOPQRSTUVWXYZabcdefghijklmnopqrstuvwxyz}
\@abbrev{scr}{\mathscr}{
	ABCDEFGHIJKLMNOPQRSTUVWXYZabcdefghijklmnopqrstuvwxyz}
\@abbrev{sf}{\mathsf}{ABCDEFGHIJKLMNOPQRSTUVWXYZabcdefghijklmnopqrstuvwxyz}
\@abbrev{b}{\bar}{abcdefghijklmnopqrstuvwxyz}
\makeatother

\newcommand{\FPC}{\mbox{\sc FPC}}

\newcommand{\CPT}{\mbox{\sc CPT}}

\newcommand{\classmonPCx}[1]{\ensuremath{\mathbf{\mathcal{K}_{\monPCx 
				k}}}}

\newcommand{\G}{\ensuremath{\mathfrak{G}}}

\newcommand{\Aut}{\mathbf{Aut}}

\renewcommand{\phi}{\varphi}

\newcommand{\lra}{\longrightarrow}

\renewcommand{\phi}{\varphi}
\renewcommand{\theta}{\vartheta}

\renewcommand{\epsilon}{\varepsilon}

\newcommand{\EPC}{\mbox{\sc EPC}}

\newcommand{\PC}{\textrm {PC}}
\newcommand{\MC}{\textrm{MC}}

\newcommand{\DWL}{\mbox{\sc DWL}}

\newcommand{\WL}{\textrm {WL}}

\newcommand{\ptime}{\mbox{\sc Ptime}}

\newcommand{\Pairs}{\text{\upshape \textbf{NewPairs}}}

\newcommand{\Cc}{{\cal C}}

\newcommand{\Kk}{{\cal K}}

\newcommand{\Oo}{{\cal O}}
\newcommand{\Pp}{{\cal P}}

\newcommand{\Ss}{{\cal S}}

\newcommand{\Vv}{{\cal V}}

\renewcommand{\bar}{\overline}

\renewcommand{\G}{G} %TODO: DECIDE IF GROUPS SHOULD BE GAMMA OR G by default

\newcommand{\pairOp }{\texttt{addPair}}
\newcommand{\sccOp }{\texttt{contract}}
\newcommand{\createOp }{\texttt{create}}
\newcommand{\forgetOp }{\texttt{forget}}

\newcommand{\dom}{\operatorname{dom}}

\newcommand{\op}{\operatorname{op}}

\newcommand{\scc}{{\operatorname{scc}}}
\newcommand{\sccs}{{\operatorname{SCC}}}
\newcommand{\pair}{{\operatorname{pair}}}
\newcommand{\diag}{{\operatorname{diag}}}

\newcommand{\Pair}{\textsf{\upshape Pair}}
\newcommand{\Union}{\textsf{\upshape Union}}

\newcommand{\cutout}[1]{}

\newcommand{\Piso}{P_{\text{iso}}}

\title{ Finite Model Theory and Proof Complexity revisited: Distinguishing graphs in Choiceless Polynomial Time and the Extended Polynomial Calculus   } %TODO Please add

\titlerunning{ Distinguishing graphs in CPT and the Extended Polynomial Calculus      }  %TODO optional, please use if title is longer than one line

\author{Benedikt Pago}{Mathematical Foundations of Computer Science, RWTH Aachen University, Germany }{pago@logic.rwth-aachen.de}{}{}%TODO mandatory, please use full name; only 1 author per \author macro; first two parameters are mandatory, other parameters can be empty. Please provide at least the name of the affiliation and the country. The full address is optional


\authorrunning{B. Pago} %TODO mandatory. First: Use abbreviated first/middle names. Second (only in severe cases): Use first author plus 'et al.'

\Copyright{Benedikt Pago} %TODO mandatory, please use full first names. LIPIcs license is "CC-BY";  http://creativecommons.org/licenses/by/3.0/

\acknowledgements{I would like to thank Daniel Wiebking for extensively answering my numerous questions on Deep Weisfeiler Leman and the properties of coherent configurations. }%optional

\nolinenumbers

\EventEditors{Bartek Klin and Elaine Pimentel}
\EventNoEds{2}
\EventLongTitle{31st EACSL Annual Conference on Computer Science Logic (CSL 2023)}
\EventShortTitle{CSL 2023}
\EventAcronym{CSL}
\EventYear{2023}
\EventDate{February 13--16, 2023}
\EventLocation{Warsaw, Poland}
\EventLogo{}
\SeriesVolume{252}
\ArticleNo{8}

\begin{document}
	
	\maketitle
	
	\begin{abstract}
		This paper extends prior work on the connections between logics from finite model theory and propositional/algebraic proof systems. We show that if all non-isomorphic graphs in a given graph class can be distinguished in the logic \emph{Choiceless Polynomial Time} with counting (CPT), then they can also be distinguished in the \emph{bounded-degree extended polynomial calculus} (EPC), and the refutations have roughly the same size as the resource consumption of the CPT-sentence. 
		This allows to transfer lower bounds for EPC to CPT and thus constitutes a new potential approach towards better understanding the limits of CPT. A super-polynomial EPC lower bound for a $\ptime$-instance of the graph isomorphism problem would separate CPT from $\ptime$ and thus solve a major open question in finite model theory.\\ 
		Further, using our result, we provide a model theoretic proof for the separation of bounded-degree polynomial calculus and bounded-degree \emph{extended} polynomial calculus.
	\end{abstract}

	\section{Introduction and results}
	In recent years, a close connection between propositional proof complexity and finite model theory has been discovered and investigated -- this is fruitful in particular because it allows the transfer of lower bounds between the two fields. 
	In \cite{berkholz}, Berkholz and Grohe showed that, with respect to the \emph{graph isomorphism problem}, \emph{fixed-point logic with counting} (FPC) has the same expressive power as the \emph{bounded-degree monomial calculus}. In \cite{ggpp} it was shown more generally that there are mutual simulations between different variants of fixed-point logic and resolution/monomial calculus, not only for the graph isomorphism problem, but for deciding any classes of finite structures. These simulations preserve the relevant complexity parameters: The number of variables in a fixed-point sentence is reflected in the width/degree of the corresponding resolution/monomial calculus refutation, and vice versa. Therefore, known lower bounds for these respective parameters can be transferred between proof complexity and finite model theory. We extend this line of research from the rather well-understood fixed-point logics to a stronger model of computation, \emph{Choiceless Polynomial Time} (CPT). This logic is an extension of FPC with a mechanism to construct (isomorphism-invariant) higher-order objects, i.e.\ nested sets, over the input structure. This power to create new objects puts CPT in a realm beyond formalisms with bounded variable number, bounded width, or bounded degree. As it turns out, a proof system that can naturally simulate this mechanism is the bounded-degree \emph{polynomial calculus with extension axioms} over $\bbQ$. \emph{Extension axioms} can be added to any proof system; they allow to introduce new variables in a proof as abbreviations for more complex expressions, which generally allows for considerably shorter proofs.\\
	We consider a similar setting as in \cite{berkholz}, that is, we compare the logic and the proof system with respect to their power to \emph{distinguish non-isomorphic graphs}. We say that CPT distinguishes all graphs in a graph class $\Kk$ if there exists a polynomial resource bound $p(n)$ such that for all pairs of non-isomorphic graphs $G,H \in \Kk$, there exists a CPT-sentence with time and space bound $p(n)$ that evaluates to true in one of the graphs and false in the other one (\autoref{def:distinguishingInCPT}). A proof system distinguishes $G$ and $H$ if it can refute the statement ``$G$ and $H$ are isomorphic'', encoded in a natural way as a propositional formula/system of polynomial equations $\Piso(G,H)$ (\autoref{def:PC_isoSystem}). Our main result reads as follows:
	\begin{theorem}
		\label{thm:PC_mainInformal}
		Let $\Kk$ be a class of graphs such that $\CPT$ distinguishes all graphs in $\Kk$. Then the degree-$3$ extended polynomial calculus over $\bbQ$ (denoted $\EPC_3$) distinguishes all graphs in $\Kk$ with refutations of polynomial size.\\
		Moreover, the $\EPC_3$-refutation uses only extension axioms $\frac{}{X_f - f}$ for polynomials of the form $f = X \cdot Y$ or $f = \frac{1}{n} \cdot \Big(\sum_{i=1}^{n^2} X_i \Big)$. That is, only monomials and certain ``averaged sums'' are replaced with new variables.
	\end{theorem}
	This has two main consequences. Most importantly, it establishes a new potential approach for the difficult open problem of proving strong lower bounds for CPT. A central topic in finite model theory is the quest for a logic that captures $\ptime$ (see \cite{chandra1980structure}, \cite{grohe2008quest}, \cite{gurevich1985logic}, \cite{pakusa2015linear}). At the moment, CPT is arguably the most prominent candidate logic for this, after rank logic has been ruled out \cite{lichter}.
	That is, evaluating any fixed CPT-sentence in a given input structure is in \ptime, and as of yet, no decision problem in $\ptime$ is known that cannot be defined by a CPT-sentence. However, the isomorphism-invariance of CPT is a severe limitation. Intuitively, it means that every classical algorithm involving choices or ordered iterations, such as e.g.\ Gaussian elimination, has to be executed in parallel for all possible orderings of the input structure, at least if it is implemented in the naive way in CPT. This requires exponential space and time resources. Thus, CPT can only be equal to P if there exists some clever trick that allows to simulate ordered iterations in a symmetry-invariant way. 
	One quite well-studied problem that is conjectured to be hard for CPT is solving linear equation systems over finite fields -- particularly hard instances of this problem arise as encodings of the isomorphism problem of \emph{Cai-Fürer-Immerman graphs} \cite{caifurimm92} or \emph{multipedes} \cite{multipedes}. Thus, if CPT does not capture \ptime, then it is quite likely that the graph isomorphism problem on a suitable graph class is a witness for that. Unfortunately, only few techniques for proving limitations of CPT are known (essentially the symmetry-based ones employed in \cite{dawar2008}, \cite{rossman2010choiceless}, \cite{pago}).
 	\autoref{thm:PC_mainInformal} opens up a new perspective, as it enables us to transfer proof-theoretic lower bounds to CPT:
	\begin{theorem}
		\label{thm:PC_mainLowerBound}
		If there is a class $\Kk$ of graphs on which the isomorphism problem is in \ptime, but which cannot be distinguished in $\EPC_3$ with polynomial-size refutations using only extension axioms of the form mentioned in \autoref{thm:PC_mainInformal}, then $\CPT \neq \ptime$.
	\end{theorem} 
	Interesting candidate graph classes are the said CFI-graphs or multipedes; crucially, without any order relation, because on certain ordered versions of these graphs, the isomorphism problem is already known to be in CPT (\cite{dawar2008}, \cite{multipedeIsomorphism}). Recently, a super-polynomial lower bound for $\EPC$ was found \cite{alekseev}. It concerns the \emph{bit-value principle} and is based on the bit-complexity of the coefficients required for a refutation. It might be a starting point in the search for graph isomorphism lower bounds, even though it seems that its proof is not directly adaptable to this problem. A second consequence of \autoref{thm:PC_mainInformal}, together with known results from finite model theory and proof complexity (\cite{dawar2008}, \cite{berkholz}), is the separation of the bounded-degree polynomial calculus and $\EPC_3$:
	\begin{theorem}
		\label{thm:PC_extensionStronger}
		There exists a sequence of pairs of non-isomorphic graphs $(G_n,H_n)_{n \in \bbN}$ such that $\Piso(G_n,H_n)$ has a polynomial-size refutation in the degree-$3$ extended polynomial calculus (using only extension axioms of the aforementioned form) but there is no $k \in \bbN$ such that the degree-$k$ polynomial calculus can refute $\Piso(G_n,H_n)$ for all $n$.
	\end{theorem}
	To our knowledge, the separation of these two bounded-degree proof systems has not explicitly been stated before -- specifically for the graph isomorphism problem. In the unbounded-degree setting, an exponential separation between $\PC$ and $\EPC$ is known, even if one only allows extension axioms of the form $\bar{X} = 1 - X$, as in \emph{polynomial calculus resolution} \cite{massimoPCR}.
	The main value of \autoref{thm:PC_extensionStronger} is that it demonstrates how finite-model-theoretic lower and upper bounds can directly lead to corresponding results in proof complexity. Other examples of finite-model-theoretic proofs for results in proof-complexity were given in \cite{ggpp}.

	\section{Preliminaries}
	All structures in this article are finite and relational. Formally, we assume that all relations are binary (whenever we need unary relations, we encode them as binary relations). We use the words ``graphs'' and ``binary structures'' interchangeably, so in particular, graphs can be vertex- or edge-coloured.  
	For a $\tau$-structure $A$ and relation symbol $R \in \tau$, $R(A)$ denotes the corresponding relation in the structure $A$. The universe of $A$ is denoted $V(A)$. We need the following concepts from finite model theory:\\
	\\
	\textbf{Weisfeiler Leman algorithm.} The $k$-dimensional Weisfeiler Leman algorithm ($k$-WL) is an incomplete graph isomorphism test that computes a canonical colouring of the $k$-tuples of vertices. Two graphs $G$ and $H$ are distinguished by $k$-WL if there is a colour class whose size is different in the colouring of $G$ and of $H$. For a precise definition and a survey, see e.g.\  \cite{sandraSiglog}.\\
	\\
	\textbf{$k$-variable counting logic.} We denote by $\Cc^k$ the $k$-variable fragment of first-order logic augmented with counting quantifiers $\exists^{\geq i}$, for every $i \in \bbN$. For two structures $G,H$ we write $G \equiv_{\Cc^k} H$ if $G$ and $H$ satisfy exactly the same $\Cc^k$-sentences. It is well-known (see Theorem 2.2 in \cite{sandraSiglog}) that $k$-WL distinguishes $G$ and $H$ if and only if $G \not\equiv_{\Cc^{k+1}} H$. In fact, the colour classes of the stable $k$-WL colouring correspond to the $\Cc^{k+1}$-types of the $k$-tuples. In this paper, we are mainly concerned with $2$-WL and $\Cc^3$-types of vertex-pairs in graphs. The $\Cc^3$-type of a pair $(v,w)$ in a structure $A$ is the collection of all $\Cc^3$-formulas $\phi(x,y)$ such that $A \models \phi(v,w)$. It contains a lot of (non-local) information, e.g.\ whether or not $v$ and $w$ are connected, the length of the shortest path between them, etc.\\
	\\
	\textbf{The bijective $k$-pebble game.} This game is played by two players, Spoiler and Duplicator, on two structures $G$ and $H$. A position of a play is a set of pebble-pairs $\pi \subseteq V(G) \times V(H)$ with $|\pi| \leq k$. In every round, Spoiler selects a subset $\pi’ \subseteq \pi$ with $|\pi'| < k$ of the current pebbles, which remain on the board. Duplicator then specifies a bijection $f : V(G) \lra V(H)$. Spoiler chooses a $v \in V(G)$, leading to the new position $\pi' \cup \{(v,f(v))\}$. Spoiler wins if the pebbles do not induce a local isomorphism between the pebbled substructures. Duplicator wins if she can play infinitely avoiding the pebbling of non-isomorphic substructures. Spoiler has a winning strategy for the bijective $k$-pebble game on $G$ and $H$ if and only if $G \not\equiv_{\Cc^k} H$ \cite{hella1996logical}.\\
	\\
	\textbf{Fixed-point logic with counting (FPC).} FPC is a standard logic of reference in algorithmic model theory. For the purposes of this paper, it suffices to know that for every sentence $\psi \in \FPC$, there is a $k$ such that whenever it holds $G \equiv_{\Cc^k} H$ for two structures, then $G \models \psi$ if and only if $H \models \psi$.
	See \cite{dawar2015nature} for a survey on this logic and its expressive power.

	\section{Choiceless Polynomial Time}
	\label{sec:CPT}
	By CPT we always mean Choiceless Polynomial Time \emph{with counting}. For details and various ways to define CPT formally, we refer to the literature: A concise survey can be found in \cite{gradel2015polynomial}. The work that originally introduced CPT as an abstract state machine model is \cite{blass1999}; later, more ``logic-like'' presentations of CPT were invented, such as Polynomial Interpretation Logic (see \cite{grapakschalkai15}, \cite{svenja}) and BGS-logic \cite{rossman2010choiceless}. In short, CPT is FPC plus a mechanism to construct isomorphism-invariant hereditarily finite sets of polynomial size. When a CPT-sentence $\Pi$ is evaluated in a finite structure $A$, then $\Pi$ may augment $A$ with hereditarily finite sets over its universe. The total number of distinct sets appearing in them (i.e.\ the sum over the sizes of the transitive closures of the h.f.\ sets) and the number of computation steps is bounded by $p(|A|)$, where $p(n)$ is a polynomial that is explicitly part of the sentence $\Pi$. We also write $\CPT(p(n))$ for the set of all CPT-sentences whose polynomial bound is at most $p(n)$. For the sake of illustration, we sketch the definition of BGS-logic:\\
	
	The sentences of BGS-logic are called \emph{programs}. A program is a tuple $\Pi = (\Pi_{\text{step}}(x),\\ \Pi_{\text{halt}}(x),\Pi_{\text{out}}(x),p(n))$. Here, $\Pi_{\text{step}}(x)$ is a BGS-term, $\Pi_{\text{halt}}$ and $\Pi_{\text{out}}$ are BGS-formulas, and $p(n)$ is a polynomial that bounds the time and space used by the program. BGS-terms take as input hereditarily finite sets and output a hereditarily finite set. Examples of such terms are $\Pair(x,y)$, which evaluates to $\{x,y\}$, or $\Union(x) = \bigcup_{y \in x} y$. Furthermore, if $s$ and $t$ are terms, $x$ is a variable, and $\phi$ a formula, then $\{ s(x) \ : \ x \in t \ : \ \phi(t) \}$ is a comprehension term. It applies the term $s$ to all elements of the set defined by $t$ that satisfy $\phi$, and outputs the set of the resulting objects $s(x)$. When a program is evaluated in a given finite structure $A$, then the term $\Pi_{\text{step}}(x)$ is iteratively applied to its own output, starting with $x_0 = \emptyset$. The iteration stops in step $i$ if the computed set $x_i = (\Pi_{\text{step}})^i(\emptyset)$ satisfies $A \models \Pi_{\text{halt}}(x_i)$. The formula $\Pi_{\text{out}}$ defines, in dependence of $x_i$, whether the run is accepting or rejecting, that is, whether $A \models \Pi$ or not. If the length of the run or the size of the transitive closure of $x_i$ exceeds $p(|A|)$ at some point, then the computation is aborted, and $A \not\models \Pi$.\\
	
	Recently, the computation model \emph{Deep Weisfeiler Leman} (DWL) has been introduced by Grohe, Schweitzer and Wiebking \cite{dwl}. DWL and CPT mutually simulate each other, and DWL is better suited to establish the connection to proof complexity.
	We present DWL in detail in Section \ref{sec:dwl}, since our proof of \autoref{thm:PC_mainInformal} actually goes via DWL.
	When we say that CPT \emph{distinguishes} certain graphs, we formally mean this:
	\begin{definition}[Distinguishing relational structures in CPT]
		\label{def:distinguishingInCPT}
		Let $\Kk$ be a class of $\tau$-structures. We say that $\CPT$ distinguishes all structures in $\Kk$ if there exists a polynomial $p(n)$ and a constant $k \in \bbN$ such that for any two structures $G, H \in \Kk$ which are non-isomorphic, there exists a sentence $\Pi \in \CPT(p(n))$ with $\leq k$ variables such that $G \models \Pi$ and $H \not\models \Pi$. 
	\end{definition}	
	This definition is perhaps non-standard because we allow the distinguishing sentence to be different for every pair of graphs in $\Kk$, whereas normally, one would expect a single sentence that distinguishes all graphs in the class. Our definition, however, matches the situation in proof complexity. As we explain in the next section, a proof system distinguishes all graphs in $\Kk$, if there exists an efficient proof for non-isomorphism of each pair of non-isomorphic graphs. This proof can of course be a different one for each pair of graphs, and the distinguishing CPT-sentences will play the role of the non-isomorphism proofs. The constant bound on the variable number is needed because with an unbounded number of variables, we could already find a first-order sentence for every graph that describes it up to isomorphism. Later on, in the DWL framework, this variable bound becomes irrelevant because DWL algorithms naturally correspond to bounded-variable CPT programs.\\
	
	The distinguishing power of CPT is strictly greater than that of FPC. This can for example be seen by considering CFI-graphs over linearly ordered \emph{base graphs} (every CFI-graph is obtained by applying the construction from \cite{caifurimm92} to a given connected base graph). We can summarise the situation like this:
	%	\newpage
	\begin{theorem}[\cite{dawar2008}]
		\label{thm:separationCPTandFPC}
		There is a  family of pairs of graphs $(G_n,H_n)_{n \in \bbN}$, that are equipped with a total preorder on the vertex set (encoded as a binary relation $\preceq$), such that:
		\begin{itemize}
			\item $G_n \not\cong H_n$ for all $n \in \bbN$. 
			\item For every $\FPC$-sentence $\psi$ and all large enough $n \in \bbN$: $\G_n \models \psi$ if and only if $H_n \models \psi$.
			\item There is a $\CPT$-sentence $\Pi$ such that for all $n \in \bbN$: $G_n \models \Pi$ and $H_n \not\models \Pi$.
		\end{itemize}	
	\end{theorem}

	\section{The (extended) polynomial calculus}
	
	\label{sec:proofSystems}
	The \emph{polynomial calculus} ($\PC$) was introduced in \cite{groebner}. It is applicable to the following problem: Given a set $P$ of multivariate polynomials over a fixed field (in our case, $\bbQ$), decide if the polynomials in $P$ have a common zero with respect to $\{0,1\}$-assignments. The polynomial $1$ is derivable from $P$ if and only if the polynomials in $P$ have no common zero over $\{0,1\}$. A derivation of the $1$-polynomial is formally a sequence $p_1,p_2,...,p_n = 1$ of polynomials such that each $p_i$ is either in $P$ or an axiom of the polynomial calculus or is obtained from one or multiple $p_j$, for $j < i$, with the application of one of the derivation rules listed below. A derivation of the $1$-polynomial from $P$ is called a \emph{refutation} of $P$.\\
	A restricted variant of the polynomial calculus, the monomial calculus, has been introduced in \cite{berkholz}.
	The derivation rules of the polynomial/monomial calculus are the following:
	\begin{definition}[Inference rules of the (extended) polynomial calculus]
		Let $P$ be the set of input polynomials/axioms, $a,b \in \bbQ$, $X$ a variable, and $f,g$ polynomials with rational coefficients.
		\begin{align*}
			&\frac{p \in P}{p} \text{ (Axioms)}  \qquad &\frac{}{X^2 - X} &\text{ (Boolean axioms)} \\ &\frac{f}{Xf} 
			\text{ (Multiplication rule)} &\qquad
			\frac{g \ \ f}{ag+bf} &\text{ (Linear 
				combination rule)} 
		\end{align*}
		In the \emph{extended polynomial calculus} ($\EPC$), \emph{extension axioms} of the form $\frac{}{X_f - f}$ may be used whenever $X_f$ is a fresh variable not occurring in $f$. The Boolean axioms do not apply to these extension variables.
	\end{definition}
	The \emph{monomial calculus} ($\MC$) is a restriction of $\PC$ that permits the
	use of the multiplication rule only in the cases where $f$ is either a 
	monomial or the product of a monomial and an axiom.	For $\MC$, $\PC$, and $\EPC$, we also consider the degree-$k$ restrictions denoted $\MC_k$, $\PC_k$, and $\EPC_k$. Proofs in these degree-restricted calculi may only consist of polynomials of degree at most $k$. In general, these proof systems are not complete any more, but bounding the degree by a constant yields natural fragments that admit efficient proof search via Gröbner basis computation (at least for $\MC_k$ and $\PC_k$, this is the case; see \cite{groebner}).\\
	The \emph{size} of a refutation $p_1,p_2,...,p_n = 1$ is the total number of occurrences of monomials in all its polynomials. Its \emph{bit-complexity} is the maximum number of bits required to represent any of the occurring coefficients, where values in $\bbQ$ are stored as a fraction of two binary numbers.\\ 

	Intuitively, the effect of the extension axioms in this setting is that the degree-bound of three may be ``locally'' violated: Monomials like $A \cdot B \cdot C \cdot D$ can be written as $X_{AB} \cdot X_{CD}$, where $X_{AB}$ and $X_{CD}$ are fresh extension variables such that $X_{AB} = A \cdot B$, and $X_{CD} = C \cdot D$. Thus, with the help of extension axioms, we can implicitly use monomials of larger degree than allowed. If we restrict ourselves to refutations of polynomial size, then we can think of $\EPC_3$ as a version of degree-$3$ polynomial calculus where the degree bound can be violated a limited number of times.

	\section{Applying algebraic proof systems to the graph isomorphism problem}
	\label{sec:PC_groheBerkholz}
	Let $G$ and $H$ be fixed graphs, potentially with a colouring of the vertices or with multiple edge relations. We consider the following polynomial axiom system $\Piso(G,H)$ that expresses the existence of a (colour-preserving) isomorphism between $G$ and $H$. A refutation of $\Piso(G,H)$ in any variant of the polynomial calculus then witnesses that $G$ and $H$ are non-isomorphic.
	This definition of $\Piso(G,H)$ is almost the same as in \cite{berkholz}.
	\begin{definition}[$P_{iso}(G,H)$, \cite{berkholz}]
		\label{def:PC_isoSystem}
		Let $G$ and $H$ be two graphs (potentially vertex-coloured). Let $\sim \subseteq V(G) \times V(H)$ be the relation "vertex $v \in V(G)$ and $w \in V(H)$ have the same colour".\\
		The system $\Piso(G,H)$ consists of the following polynomials in the variables $\{X_{vw} \mid v \in V(G), w \in V(H), v \sim w\}$.
		\begin{align}
			\label{isoAxiom1}
			&\sum_{\stackrel{v \in V(G)}{v \sim w}} X_{vw} - 1 & &\text{ for all } w \in V(H)
		\end{align}
		\begin{align}
			\label{isoAxiom2}
			&\sum_{\stackrel{w \in V(H)}{v \sim w}} X_{vw} - 1 & &\text{ for all } v \in V(G)
		\end{align}
		\begin{align}
			\label{isoAxiom3}
			&X_{vw}X_{v'w'} & &\text{ for all } v,v' \in V(G), w,w' \in V(H)\\
			& & &\text{ with } v \sim w \text{ and } v' \sim w' \notag\\ 
			& & &\text{ such that } \{(v,w),(v',w')\} \text{ is not} \notag\\
			& & & \text{ a local isomorphism.} \notag
		\end{align}
	\end{definition}
	The intended meaning of the variable $X_{vw}$ being set to one is ``$v$ is mapped to $w$''. 
	When we say that a certain variant of the polynomial calculus \emph{distinguishes} two graphs $G,H$, we mean that the polynomial equation system $\Piso(G,H)$ has a refutation in that proof system.
	The main result from \cite{berkholz} links graph distinguishability in this sense to graph distinguishability by the $k$-dimensional Weisfeiler Leman algorithm. 
	\begin{theorem}[Theorem 4.4 in \cite{berkholz}]
		\label{thm:PC_mainGroheBerkholz}
		Let $k \in \bbN$ and let $G$ and $H$ be graphs.
		The axiom system $\Piso(G,H)$ has a refutation in the degree-$k$ monomial calculus iff the $(k-1)$-dimensional Weisfeiler Leman algorithm distinguishes $G$ and $H$.
	\end{theorem}
	In \cite{berkholz}, this is not stated for vertex- or edge-coloured graphs, but it can be checked that the proof still goes through in these cases.
	For our result, we need some of the technical ingredients from the proof of \autoref{thm:PC_mainGroheBerkholz}: What is shown in \cite{berkholz} is that Spoiler's winning positions in the bijective $k$-pebble game on $G$ and $H$ are derivable in $\MC_k$ from $\Piso(G,H)$.\\
	A position in the game is a set of pebble pairs $\pi \subseteq V(G) \times V(H)$ of size $|\pi| \leq k$. The position $\pi$ corresponds to a monomial in the variables from $\Piso(G,H)$. We denote this monomial as $X_\pi := \prod_{(v,w) \in \pi} X_{vw}$ (so the $X_{vw}$ are the variables, whereas $X_\pi$ is shorthand for a \emph{product} of variables). We will use the following central technical result as a blackbox:
	\begin{lemma}[Lemma 4.2 in \cite{berkholz}]
		\label{lem:PC_groheberkholzlemma}
		Let $k \geq 2$ and $G,H$ be graphs (such that for every vertex-colour $Q$, there are exactly as many vertices of colour $Q$ in $G$ as in $H$). If Spoiler has a winning strategy for the bijective $k$-pebble game on $G,H$ with initial position $\pi$, then there is an $\MC_k$-derivation of the monomial $X_{\pi}$ from $\Piso(G,H)$.
	\end{lemma}
	This lemma accounts for one direction of \autoref{thm:PC_mainGroheBerkholz} because if $G \not\equiv_{\Cc^{k}} H$, then Spoiler wins the bijective $k$-pebble game from the initial position $\emptyset$, and we have $X_\emptyset = 1$.

	\section{Separating the bounded-degree extended polynomial calculus from its non-extended version}
	\label{sec:PC_separation}
	Before we come to the more technical part, we show how \autoref{thm:PC_extensionStronger} follows from \autoref{thm:PC_mainInformal}. According to Theorem 6.2 in \cite{berkholz}, for every $n \in \bbN$, there exist pairs $G_n, \widetilde{G}_n$ of non-isomorphic CFI-graphs of size $\Oo(n)$ such that $\Piso(G_n, \widetilde{G}_n)$ has no degree-$n$ polynomial calculus refutation (over $\bbQ$). A closer examination of the construction in \cite{berkholz} reveals that the axioms \eqref{isoAxiom1} and \eqref{isoAxiom2} in $\Piso(G_n, \widetilde{G}_n)$ are for coloured versions of the respective CFI-graphs: Each vertex-gadget and each edge-gadget of $G_n$ and $\widetilde{G}_n$, respectively, forms a distinct colour class and the axioms restrict possible isomorphisms to colour-preserving ones (the precise definition of the said vertex- and edge-gadgets is not essential here, so we refer to \cite{dawar2008} for the presentation of the CFI construction). 
	Therefore, we have $\Piso(G_n, \widetilde{G}_n) = \Piso((G_n, \preceq), (\widetilde{G}_n, \preceq))$ for any preorder $\preceq$ on $V(G_n)$ that is obtained from a linear order on the respective base graph; in other words, a preorder that linearly orders the CFI-gadgets without ordering the vertices inside each gadget (for details, see \cite{caifurimm92} or \cite{dawar2008}). Note that our definition of $\Piso$ also works for graphs with multiple edge relations, and we can simply view the binary relation $\preceq$ as another type of edge relation.
	Now the system $\Piso((G_n, \preceq), (\widetilde{G}_n, \preceq))$ does have a polynomial-size refutation in $\EPC_3$, for any choice of the ordering, because CFI-graphs over linearly ordered base graphs can be distinguished in CPT \cite{dawar2008} and thus, a refutation exists by our \autoref{thm:PC_mainInformal}.
	
	\section{Deep Weisfeiler Leman}
	\label{sec:dwl}
	Before we are ready to prove \autoref{thm:PC_mainInformal}, we have to introduce the technical details of \emph{Deep Weisfeiler Leman}, a computation model equivalent to CPT that we will simulate in $\EPC_3$.\\
	A DWL-algorithm is a deterministic Turing machine that gets as input a finite structure with binary relations. All DWL-computations are isomorphism-invariant: The machine does not have access to the input structure directly, but only to its so-called \emph{algebraic sketch}. This is a certain invariant of the structure, similar to its $2$-dimensional Weisfeiler Leman colouring. The machine is not only able to read information about its input structure but it can also modify the structure in an isomorphism-invariant way. These modifications correspond to the creation of higher-order objects in Choiceless Polynomial Time.\\
	Before we can introduce the Deep Weisfeiler Leman framework in detail, we have to say precisely what the algebraic sketch of a structure is. It is a representation of its \emph{coarsest coherent configuration} (or ``coherent colouring''), that is defined below. The coarsest coherent configuration of a structure is also known as its stable $2$-dimensional Weisfeiler Leman colouring (equivalent to the partition of all pairs into their $\Cc^3$-types). 
	The following presentation closely follows the one in \cite{dwl}. 
	
	\subsection{Coherent configurations of binary structures}
	\label{sec:preliminariesDWL}
	
	\begin{definition}[Notions concerning binary relations, \cite{dwl}]
		~\\
		\vspace{-\baselineskip}
		\label{def:PC_binaryRelations}	
		\begin{itemize}
			\item The \emph{converse} of a relation $R$ is the relation $R^{-1}:=\{(v,u)\mid
			(u,v)\in R\}$. 
			\item For a set $V$, the \emph{diagonal} of $V$ is the
			relation $\diag(V):=\{(v,v)\mid v\in V\}$. For a relation $R \subseteq V^2$ we let
			$R^\diag:=R\cap\diag(V)$ be the diagonal elements in~$R$. We call $R$ a \emph{diagonal relation} if $R=R^\diag$.
			\item The \emph{strongly connected components} (\emph{SCC}s) of a relation
			$R$ are defined in the usual way as inclusionwise
			maximal sets $S$ such that for all $u,v\in S$ there is an
			$R$-path of length at least $1$ from $u$ to $v$.
			(A singleton
			set $\{u\}$ can be a strongly connected component only if $(u,u)\in R$.)
			We write $\scc(R)$ to denote the set of strongly connected
			components of $R$. Moreover, we let $R^\scc:=\bigcup_{S\in\scc(R)}S^2$ be the relation describing whether two elements are in the same strongly connected component.
		\end{itemize}
	\end{definition}
	
	\begin{definition}[Coherent configurations, \cite{dwl}]
		\label{def:PC_coherentConfig}
		Let $\sigma$ be a vocabulary.
		A \emph{coherent $\sigma$-configuration} $C$ is a $\sigma$-structure $C$
		with the following properties.
		\begin{itemize}
			\item $\{R(C) \mid R\in\sigma \}$ is a partition of $V(C)^2$.
			\item For each $R \in\sigma$ the relation $R(C)$ is either a
			subset of or disjoint from the diagonal $\diag(V(C))$.
			\item For each $R\in\sigma$ there is an
			$R^{-1}\in\sigma$ such that ${R^{-1}}(C)= (R(C))^{-1}$.
			\item For all $R_1,R_2,R_3\in\sigma$ there is a number
			$q = q(R_1, R_2 ,R_3)\in\bbN$ such that for all $(u,v)\in R_1(C)$ there are
			exactly $q$ elements $w \in V(C)$ such that $(u,w)\in R_2(C)$ and $(w,v)\in
			R_3(C)$.
		\end{itemize}
		The numbers $q(R_1,R_2,R_3)$ are called the \emph{intersection numbers} of
		$C$ and the function $q:\sigma^3\to \bbN$ is called the
		\emph{intersection function}.\\
	\end{definition}
	\vspace{-0.2cm}
	A coherent $\sigma$-configuration $C$ is at least as
	\emph{fine} as, or \emph{refines}, a
	$\tau$-structure $A$ (we write $C \sqsubseteq A$ ) if $V(A) = V(C)$, and for each $R\in\sigma$
	and each $E\in\tau$
	it holds that $R(C) \subseteq E(A)$ or that
	$R(C)\subseteq A^2\setminus E(A)$.
	% there is a set $\pi\subseteq\sigma$ such that
	% $E(A)=\bigcup_{R\in\pi}R(C)$.
	We say that a coherent configuration $C$ is a \emph{coarsest coherent configuration refining} a
	structure $A$ if $C \sqsubseteq A$ and $C'\sqsubseteq C$ for every coherent
	configuration $C'$ satisfying~$C'\sqsubseteq A$. In the following, we will usually write $\tau$ for the vocabulary of a given structure and $\sigma$ for the vocabulary of the corresponding coherent configuration, without further specifying $\sigma$.

	Every binary structure $A$ has a coarsest coherent configuration refining it, which can be computed efficiently with the $2$-dimensional Weisfeiler Leman algorithm (Theorem 2.1 in \cite{dwl}). This configuration is unique up to the renaming of relation symbols. We write $C(A)$ for the coarsest coherent configuration of $A$ with canonical names of the relation symbols, as for example produced by a fixed implementation of $2$-WL. 
	We call the relation symbols in $\sigma$ \emph{colours} to distinguish them from the relation symbols in $\tau$. In the following, we often identify the symbols in $\tau$ and $\sigma$ with binary strings, because this is how they are represented in a Turing machine.

	The \emph{algebraic sketch} of a structure contains information about the colours appearing in its coarsest coherent configuration, which relations of the structure they refine, and the intersection function $q$. Formally, the algebraic sketch of a structure $A$
	is the tuple $D(A) = (\tau,\sigma, \subseteq_{\sigma,\tau},q)$. The relation $\subseteq_{\sigma,\tau}$ relates the colours in $\sigma$ with the relations in $\tau$ they refine: $\subseteq_{\sigma,\tau} := \{ (R,E) \in \sigma \times \tau \mid R(C(A)) \subseteq E(A)  \}$. 
	In order to feed $D(A)$ to a Turing machine, we have to agree on some encoding in binary. If the binary string encodings of the relation symbols are fixed, then there is a canonical encoding of $D(A)$, based on the lexicographic ordering of the relation names and ordering of the intersection numbers  $q(R_1,R_2,R_3)$. The string that encodes $D(A)$ is the initial tape content in a DWL-computation on the structure $A$.
	
	\subsection{The Deep Weisfeiler Leman computation model} 
	\label{sec:PC_DWLdefinition}
	A DWL-algorithm is a two-tape Turing machine $M$ with an additional storage device that the authors of \cite{dwl} have named ``the \emph{cloud}''. It contains a structure $A$ together with its coarsest coherent configuration $C(A)$.
	The storage that the machine itself can use is a \emph{work tape} and an \emph{interaction tape}, which allows for interaction with the cloud. The input of a DWL-program $M$ is a binary $\tau$-structure $A$. Initially, the cloud contains the coherently coloured structure $(A, C(A))$, and on the interaction tape, the algebraic sketch $D(A)$ is written. The work tape is empty.\\
	The Turing machine works as a standard Turing machine with two special transitions that can modify the structure in the cloud. To execute these, the machine writes a binary string $s \in \{0,1\}^*$ on the interaction tape and enters one of the two distinguished states $q_{\pairOp}, q_{\sccOp}$. If $s$ is a colour $R \in \sigma$, then we say that $M$ executes $\pairOp(R)$ or $\sccOp(R)$. Executing $\pairOp(R)$ creates a new vertex for each vertex-pair whose colour in $C(A)$ is $R$. The operation $\sccOp(R)$ contracts every SCC formed by pairs of colour $R$ into a single vertex.
	\begin{itemize}
		\item $\pairOp(R)$: The machine adds a fresh vertex for each pair in $R(C(A))$ to $A$, i.e.\ it updates $V(A)$ to $V(A) \uplus R(C(A))$. These pairs are then connected with their elements in $A$. Therefore, $\tau$ is updated to $\tau \cup \{E_{\text{left}},E_{\text{right}}\} \uplus \{D_R\}$. The new relation $D_R$ identifies the newly added vertices, i.e.\ $D_R(A) := \diag(R(C(A)))$. The symbol $D_R$ is chosen as the lexicographically smallest binary string that is not yet used as a relation symbol.\\
		Furthermore, $E_{\text{left}}(A)$ is updated to $E_{\text{left}}(A) \cup \{(u,(u,v)) \in V(A)^2 \mid (u,v) \in R(C(A))  \}$, and $E_{\text{right}}(A)$ is set to $E_{\text{right}}(A) \cup \{(v,(u,v)) \in V(A)^2 \mid (u,v) \in R(C(A))  \}$ (where initially, $E_{\text{left}}(A) = E_{\text{right}}(A) = \emptyset$).
		\item $\sccOp(R)$: Let $\Ss := \scc(R)$ be the set of strongly connected components of the relation $R$. Let $U := V(A) \setminus \bigcup \Ss$ be the set of vertices that are not in one of the strongly connected components. The components in $\Ss$ are contracted. That means we update $V(A)$ to $U \uplus \Ss$, and $\tau$ to $\tau \uplus \{D_R\}$. Again, $D_R(A) := \diag(\Ss)$. For each relation $E \in \tau$, we update $E(A)$ to $(E(A) \cap U^2) \cup \{(u,S) \mid \exists v \in S \in \Ss: (u,v) \in E(A)  \} \cup \{ (S,v) \mid  \exists u \in S \in \Ss: (u,v) \in E(A) \} \cup \{ (S_1,S_2) \mid \exists u \in S_1 \in \Ss, \exists v \in S_2 \in \Ss: (u,v) \in E(A)  \}$. 
	\end{itemize}
	Each of these special transitions modifies the structure $A$ in the cloud in an isomorphism-invariant way. After that, the cloud storage device computes $C(A)$ (for example, with the $2$-WL algorithm) and stores the coherently coloured structure $(A,C(A))$. The algebraic sketch $D(A)$ of the new structure is written on the interaction tape.\\
	A DWL-algorithm decides a class $\Kk$ of $\tau$-structures in the usual sense, i.e.\ the algorithm halts with output $1$ on input $A$ if $A \in \Kk$, and else, it halts with output $0$. We say that a DWL-algorithm \emph{runs in polynomial time} if the number of computation steps of the Turing machine and the size of the structure in the cloud is bounded by a polynomial in the size of the input structure.
	The original definition of DWL in \cite{dwl} also has two more operations, $\createOp$ and $\forgetOp$, but it can be shown that these do not increase the expressive power; the proof is similar to the proof in \cite{dwl} showing that ``pure DWL'' simulates DWL, so we omit it.

	\subsection{Distinguishing graphs in Deep Weisfeiler Leman}
	\label{sec:PC_distinguishingGraphsDWL}
	
	In \cite{dwl}, the authors say that a DWL-algorithm \emph{decides isomorphism} on a structure class $\Kk$ if it gets as input the disjoint union $A:= G \uplus H$ of two connected binary structures and correctly decides whether $G \cong H$. A crucial technical result in \cite{dwl} shows that one can always assume that at any stage of the computation, the structure in the cloud is the disjoint union of two connected structures: It is never necessary for the algorithm to produce connections between the two components, i.e.\ $\pairOp$ and $\sccOp$ are only executed for colours $R$ with $R(C(A)) \subseteq V(G)^2 \cup V(H)^2$. A DWL-algorithm that maintains this invariant is called \emph{normalised} in \cite{dwl}. 
	\begin{definition}[Distinguishing structures in DWL]
		\label{def:PC_distinguishInDWL}
		The computation model $\DWL$ \emph{distinguishes} all structures in a class $\Kk$ (of connected $\tau$-structures) in \emph{polynomial time} if: There is a polynomial $p(n)$ such that for any two non-isomorphic structures $G, H \in \Kk$, there exists a \emph{normalised} DWL-algorithm $M$ which, given $G\uplus H$ as input, terminates with a structure $G' \uplus H'$ in the cloud such that $D(G') \neq D(H')$, and takes time and space at most $p(|G|+|H|)$.	
	\end{definition}	
	This is simply the DWL-version of \autoref{def:distinguishingInCPT} for distinguishing structures in CPT. The main difference to the CPT-setting is that here, the constant bound on the number of variables is already implicit in the definition of DWL (because DWL only accesses a structure via its coherent configuration, and two structures with the same configuration are $\Cc^3$-equivalent).\\
	Let us elaborate on what is meant precisely by $D(G') \neq D(H')$.
	Whenever $A = A_1 \uplus A_2$ is a $\tau$-structure consisting of two separate connected components, then we write $D(A)[A_1]$ and $D(A)[A_2]$ for the restrictions of the algebraic sketch $D(A)$ to the $\sigma$-colours that occur as colours of pairs in $V(A_1)^2$ or $V(A_2)^2$, respectively. 
	It follows from the properties of normalised DWL-computations (Lemma 8 in \cite{dwl}) that $D(A)[A_i]$ is in fact the algebraic sketch of $A_i$, so the sketch of $A_1 \uplus A_2$ is composed of the sketches of the two structures:
	\begin{lemma}
		\label{lem:restrictionOfSketch}
		Let $A = A_1 \uplus A_2$ be the disjoint union of two connected $\tau$-structures. For $i \in \{1,2\}$, $D(A)[A_i]$ is an algebraic sketch and equal to $D(A_i)$, up to a renaming of the colours in $\sigma$.
	\end{lemma}  
	Thus, when we write $D(G') \neq D(H')$, we are formally referring to the respective restrictions of $D(G' \uplus H')$, but these are equivalent to $D(G')$ and $D(H')$, respectively.
	The algebraic sketches being distinct means that $2$-WL distinguishes the structures. By standard results (see e.g.\ \cite{sandraSiglog}), one can infer:
	\begin{lemma}
		\label{lem:sketchesDistinctSpoilerWins}
		Let $G, H$ be connected $\tau$-structures. It holds $D(G) \neq D(H)$ if and only if Spoiler has a winning strategy for the bijective $3$-pebble game on $G$ and $H$. 
	\end{lemma}	
	
	We can say even more, namely that Spoiler can distinguish pairs of different colours. The following is a variation of a standard result (Theorem 2.2 in \cite{sandraSiglog}). The standard result concerns the setting where the graphs $G$ and $H$ are considered separately with their respective coarsest coherent configurations. Here, we have to work with their disjoint union. It is perhaps not surprising that in this setting, the correspondence between Weisfeiler-Leman colourings and pebble games also exists, but we are not aware of a formal proof for this statement for $G \uplus H$ in the literature. Thus, for completeness, we provide one in the appendix.
	\begin{restatable}{lemma}{restatePebbleLemma}
		\label{lem:PC_equivWLPebble}
		Let $G,H$ be two connected $\tau$-structures and $A := G \uplus H$ with its coarsest coherent $\sigma$-configuration $C(A)$. Let $(v,v') \in V(G)^2, (w,w') \in V(H)^2$ such that there is no $R \in \sigma$ with $(v,v') \in R(C(A))$ and $(w,w') \in R(C(A))$.\\
		Then Spoiler has a winning strategy for the bijective 3-pebble game on $G$ and $H$ with initial position $\{(v,w),(v',w')\}$.
	\end{restatable}
	
	As shown below, if CPT distinguishes all structures in a class $\Kk$, then also DWL distinguishes all structures in $\Kk$ in polynomial time. Hence, in our proof of \autoref{thm:PC_mainInformal}, we can indeed start with the assumption that DWL polynomially distinguishes all graphs in $\Kk$.
	\begin{lemma}
		\label{lem:connectionDWLdistinguishabilityCPT}
		Let $\Kk$ be a class of connected structures that are distinguished by $\CPT$ in the sense of \autoref{def:distinguishingInCPT}.
		Then DWL distinguishes all structures in $\Kk$ in polynomial time in the sense of \autoref{def:PC_distinguishInDWL}.
	\end{lemma}
	%\todo[inline]{TODO: only sketch this proof. Rest in appendix}
	\textit{Proof sketch.} 	Let $p(n)$ be the resource bound for the distinguishing CPT-programs for the class $\Kk$ that exists by \autoref{def:distinguishingInCPT}.
	Fix two $\tau$-structures $G, H \in \Kk$ such that $G \not\cong H$. Let $\Pi \in \CPT(p(n))$ be a distinguishing sentence. By Theorem 21 in \cite{dwl}, there exists a polynomial time DWL-algorithm $M$ which simulates $\Pi$ (and the polynomial resource bound of $M$ depends only on $p(n)$, not on $G$ and $H$). That means $M$ w.l.o.g.\ accepts $G$ and rejects $H$. The sequence of executed cloud-interaction operations from $\{ \pairOp, \sccOp  \}$ is the same in the run of $M$ on $G$ as in the run on $H$, up to the point where the respective structures in the cloud have distinct algebraic sketches (because the behaviour of $M$ only depends on the algebraic sketch of the structure in the cloud). At that point, we can stop the simulation of $\Pi$ by $M$ because we do not actually care about the acceptance behaviour as long as the machine produces distinct sketches on $G$ and $H$. Now the same sequence of cloud-interaction-operations can be simulated by an appropriate DWL-algorithm $M'$ on input $G\uplus H$, which leads to a structure $G' \uplus H'$ with $D(G') \neq D(H')$. Such an $M'$ can be constructed because of \autoref{lem:restrictionOfSketch}.  \qedsymbol

	\section{Properties of coherent configurations}
	\label{sec:coherentConfigurations}
	Here is a small collection of lemmas concerning coherent configurations. We will need them in our construction of the $\EPC_3$-refutation in the next section.
	\begin{lemma}
		\label{lem:PC_endpointsColour}
		Let $A$ be a $\tau$-structure and $C(A)$ its coarsest coherent $\sigma$-colouring. Let $R \in \sigma$. There are diagonal colours $D_1, D_2 \in \sigma$ such that for all pairs $(u,v) \in R(C(A))$ we have $(u,u) \in D_1(C(A))$, and $(v,v) \in D_2(C(A))$. 
	\end{lemma}
	\begin{proof}
		This is Corollary 2.1.7 in \cite{chen2019lectures} (the term ``fibers'' there means the same as our ``diagonal colours'').
	\end{proof}

	\begin{corollary}
		\label{cor:PC_endpointsColour}
		Let $A$ be a $\tau$-structure and $C(A)$ its coarsest coherent $\sigma$-colouring. Let $R \in \sigma$. If for any $(u,v) \in R(C(A))$, $(u,u)$ or $(v,v)$ is in some relation $E(A)$, for $E \in \tau$, then for all other $(u',v') \in R(C(A))$, it also holds $(u',u') \in E(A)$, or $(v',v') \in E(A)$, respectively.
	\end{corollary}
	\begin{proof}
		Follows from the previous lemma and the fact that the coarsest coherent configuration of $A$ is a refinement of the relations of $A$.
	\end{proof}
	
	\begin{lemma}
		\label{lem:PC_SCCsSingleColoured}
		Let $A$ be a $\tau$-structure and let $C(A)$ be its coarsest coherent \\$\sigma$-configuration. Let $R \in \sigma$ and let $\Ss = \sccs(R)$ be the set of $R$-SCCs in $C(A)$. There is a diagonal relation $P \in \sigma$ such that $\diag(\bigcup \Ss) = P(C(A))$.
	\end{lemma}
	\begin{proof}
		First, we show that $\diag(\bigcup \Ss) \subseteq P(C(A))$.	
		Any vertex $v \in \bigcup \Ss$ has some outgoing $R$-neighbour $w$ that is in the same SCC (possibly, $w = v$). That is, we have $(v,w) \in R(C(A))$. By \autoref{lem:PC_endpointsColour}, there are specific diagonal relations $D_1, D_2$ such that $(v,v) \in D_1(C(A)), (w,w) \in D_2(C(A))$, and all endpoints of $R$-edges have these diagonal colours. But since $w$ is itself the left entry in some other $R$-edge $(w,w')$, we must have $D_1 = D_2 =: P$.\\
		It remains to prove $P(C(A)) \subseteq \diag(\bigcup \Ss)$. For any $v \in \diag(\bigcup \Ss)$, there exists an $R$-path of length $\geq 1$ from $v$ to itself. We have already argued that $(v,v) \in P(C(A))$. It follows that any other vertex $w$ with $(w,w) \in P(C(A))$ also has an $R$-path to itself and is thus in $\bigcup \Ss$. To see this, recall that $C(A)$ corresponds to the stable $2$-\WL -colouring \cite{dwl}, which in turn partitions $A^2$ into $\Cc^3$-types \cite{sandraSiglog}. Hence, all vertices with the same diagonal colour satisfy exactly the same $\Cc^3$-formulas. 
		The existence of an $R$-path from a vertex to itself (in the fixed structure $A$) is expressible in $\Cc^3$ using standard techniques: Namely, for every fixed number $d \leq |A|$, we can write a $\Cc^3$ formula $\phi_d(x,y)$ that asserts the existence of a path of length $d$ from $x$ to $y$. Only $3$ variables are needed because one can alternately requantify used variables (see e.g.\ Proposition 3.2 in \cite{immermanLander}). In the formula, we have access to the relation $R$ because it is itself $\Cc^3$-definable: Essentially, $R$ is a $\Cc^3$-type of vertex-pairs in $A$, and it is known that on finite structures, such a type is definable with a single formula.
	\end{proof}

	\begin{corollary}
		\label{cor:PC_SCCsSingleVertexColoured}
		Let $A$ be a $\tau$-structure such that for every $v \in V(A)$, $(v,v)$ is in exactly one diagonal relation $P(A)$, for $P \in \tau$ (i.e.\ $A$ is a graph with vertex colours).
		Let $C(A)$ be the coarsest coherent $\sigma$-configuration of $A$. Let $R \in \sigma$ and let $\Ss = \sccs(R)$. There is a colour (i.e.\ a diagonal relation) $P \in \tau$ such that for every SCC $S \in \Ss$, $\diag(S) \subseteq P(A)$.
	\end{corollary}	
	\begin{proof}
		The coarsest coherent configuration $C(A)$ is a refinement of $A$. Therefore, the diagonal relations in $C(A)$ are subsets of the diagonal relations in $A$. Now the statement follows directly from \autoref{lem:PC_SCCsSingleColoured}.
	\end{proof}

	\begin{lemma}
		\label{lem:PC_SCCequalSize}
		Let $A,C(A), R \in \sigma$, and $\Ss$ be as above. All $R$-SCCs in $\Ss$ are of equal size.
	\end{lemma}
	\begin{proof}
		For any number $k \in \bbN$, we can write a $\Cc^3$-formula $\phi_k(x)$ asserting that the size of the $R$-SCC of $x$ is exactly $k$. To do this, we can just use a counting quantifier and the fact that the existence of an $R$-path between two vertices (and back) is $\Cc^3$-definable (see proof of \autoref{lem:PC_SCCsSingleColoured}). Now since all vertices in $R$-SCCs have the same diagonal colour (\autoref{lem:PC_SCCsSingleColoured}), and colours coincide with $\Cc^3$-types, they all satisfy the same $\phi_k$ and thus, all SCCs have equal size.
	\end{proof}
	
	\begin{lemma}
		\label{lem:PC_differentColours}
		Let $A,C(A), R \in \sigma$ and $\Ss$ be as above.
		There is a collection of colours $R_1,...,R_t \in \sigma$ such that $R^\scc = \bigcup_{S \in \Ss} S^2 = \bigcup_{i =1}^t R_i(C(A))$.
	\end{lemma}
	\begin{proof}
		We let $R_1,...,R_t$ be the smallest collection of colours such that every $(u,v) \in \bigcup_{S \in \Ss} S^2$ occurs in one of them. To see that this has the desired property, let $(u,v) \in V(A)^2$ be such that $u$ and $v$ are not in the same SCC. 
		Let $T \in \sigma$ be the colour such that $(u,v) \in T(C(A)))$. There is no $R$-path from $u$ to $v$ and back. As already argued in the proof of \autoref{lem:PC_SCCsSingleColoured}, this fact is expressible in $\Cc^3$. Since there does exist an $R$-path in both directions between any two vertices inside each SCC, and colours coincide with $\Cc^3$-types, $T$ cannot be among the $R_i$.
	\end{proof}

	The next lemma tells us that the colour of a pair $(v',z)$ between some vertex $v'$ and any other vertex $z$ inside a given SCC contains the information whether or not there exists an edge from $v'$ into the SCC. In particular, the colour ``between $v'$ and the SCC'' is independent of the choice of the vertex $z$ in the SCC.
	This also explains why contracting SCCs is possible without loss of information.
	
	\begin{lemma}
		\label{lem:PC_coloursBetweenSCCsDeterminedByEdgesVersion1}
		Let $A,C(A), R \in \sigma$ and $\Ss$ be as above. Fix any relation symbol $E \in \tau$. Let $V,W \in \Ss$ and $v',w' \in V(A)$ be such that:\\
		There is a $v \in V$ such that $(v',v) \in E(A)$, and for all $w \in W$ it holds $(w',w) \notin E(A)$. \\
		Let $z \in V$ be arbitrary and $T \in \sigma$ such that $(v',z) \in T(C(A))$. Then $T(C(A)) \cap (\{w'\} \times W) = \emptyset$.
	\end{lemma}
	\begin{proof}
		There is a $\Cc^3$-formula $\phi(x,y)$ that asserts: There exists some vertex $y'$ in the same $R$-SCC as $y$ such that $(x,y') \in E(A)$. This formula can be constructed as described in the proof of \autoref{lem:PC_SCCsSingleColoured}. Since $\phi(x,y)$ is satisfied in $A$ for $x \mapsto v'$ and $y \mapsto z$, for any $z \in V$, but not for $x \mapsto w'$ and $y \mapsto w$, for any $w \in W$, the lemma follows again from the fact that pairs with distinct $\Cc^3$-types receive distinct colours in $C(A)$.
	\end{proof}

	The next lemma is of a similar kind. It states that the colours of pairs between different SCCs contain the information whether or not there exist edges between the two SCCs.

	\begin{lemma}
		\label{lem:PC_coloursBetweenSCCsDeterminedByEdges}
		Let $(A,C(A))$, $R \in \sigma$ and $\Ss$ be as above. Fix any relation symbol $E \in \tau$. Let $V,W,V',W' \in \Ss$ be such that: There exists $v' \in V', v \in V$ such that $(v',v) \in E(A)$, and for all $w' \in W'$, all $w \in W$, $(w',w) \notin E(A)$.\\
		Let $z \in V, z' \in V'$ be arbitrary, and let $T \in \sigma$ such that $(z',z) \in T(C(A))$. Then $T(C(A)) \cap (W' \times W) = \emptyset$.
	\end{lemma}
	\begin{proof}
		Analogous to the proof of \autoref{lem:PC_coloursBetweenSCCsDeterminedByEdgesVersion1}. Here, we use a $\Cc^3$-formula $\phi(x,y)$ that asserts: There exists $x'$ in the same SCC as $x$, and $y'$ in the same SCC as $y$ such that $(x',y') \in E(A)$.
	\end{proof}

	\section{Refuting graph isomorphism in  the extended polynomial calculus - Proof of Theorem \ref{thm:PC_mainInformal}}	
	
	\label{sec:PC_mainProof}	
	Let $\Kk$ be a class of connected binary structures such that $\CPT$ distinguishes all structures in $\Kk$.
	By \autoref{lem:connectionDWLdistinguishabilityCPT}, then also DWL distinguishes all structures in $\Kk$ in polynomial time. 
	Now \autoref{thm:PC_mainInformal} follows from \autoref{lem:PC_transferToProofComplexity} below that establishes the link between DWL-distinguishability and the extended polynomial calculus. 
	Before we can prove \autoref{lem:PC_transferToProofComplexity}, we have to state the key technical result that it depends on: 
	\begin{restatable}{lemma}{restateAxiomDerivationLemma}
		\label{lem:PC_constructionOfNewAxiomSystem}
		Let $G,H$ be two connected binary $\tau$-structures, which are potentially vertex-coloured in such a way that for every vertex-colour $Q$, there are as many vertices with colour $Q$ in $G$ as in $H$.
		Let $\op \in \{ \pairOp , \sccOp \}$ and let $R \in \sigma$, where $\sigma$ is the vocabulary of the coarsest coherent configuration $C:= C(G \uplus H)$. Assume that $R(C) \subseteq V(G)^2 \cup V(H)^2$.\\
		Let $G' \uplus H'$ be the result of executing $\op(R)$ on $G \uplus H$.\\
		\\
		Then the polynomial axiom system $\Piso(G',H')$ is derivable from $\Piso(G,H)$ in $\EPC_3$, up to a renaming of variables. The number of extension variables used in the derivation is at most $|V(G')|^2$, and the derivation has polynomial size and uses only coefficients with polynomial bit-complexity. Moreover, for every extension axiom $\frac{}{X_f -f}$ used in the derivation, $f$ is of the form $f = X \cdot Y$ or $f = \frac{1}{n} \cdot \Big(\sum_{i=1}^{n^2} X_i\Big)$.
	\end{restatable}
	
	The proof is quite lengthy and would interrupt the proof of \autoref{thm:PC_mainInformal} at this point; therefore, we first present the lemma and proof that explains how \autoref{thm:PC_mainInformal} follows from \autoref{lem:PC_constructionOfNewAxiomSystem}. Afterwards, we provide the actual polynomial calculus derivations whose existence is claimed in \autoref{lem:PC_constructionOfNewAxiomSystem}. 
	
	\begin{lemma}
		\label{lem:PC_transferToProofComplexity}
		Let $G, H$ be two connected binary $\tau$-structures. Let $p(n)$ be a polynomial and $M$ be a normalised $\DWL$-algorithm which produces on input $G \uplus H$ a structure $G' \uplus H'$ with $D(G') \neq D(H')$, such that the length of the run and the size of the structure in the cloud is bounded by $p(|G| + |H|)$ at any time.\\
		Then the system $\Piso(G,H)$ has an $\EPC_3$-refutation that uses at most $p(|G| + |H|)^3$ many extension variables, has polynomial size and polynomial bit-complexity. Moreover, for every extension axiom $\frac{}{X_f -f}$ used in the derivation, $f$ is of the form $f = X \cdot Y$ or $f = \frac{1}{n} \cdot \Big(\sum_{i=1}^{n^2} X_i\Big)$.
	\end{lemma}	
	\begin{proof}
		Follows from \autoref{lem:PC_constructionOfNewAxiomSystem} together with \autoref{lem:sketchesDistinctSpoilerWins} and \autoref{lem:PC_groheberkholzlemma}. In detail: Let $(\op_i(R_i))_{i \leq t}$ with $\op_i \in \{\pairOp, \sccOp\}$ be the sequence of cloud-interaction-operations in the run of $M$ on $G \uplus H$. This sequence of operations produces a sequence of structures $(G \uplus H, G_1 \uplus H_1, G_2 \uplus H_2, ..., G_t \uplus H_t)$ such that $D(G_t) \neq D(H_t)$.\\
		For each $i$, $R_i$ is a colour in the coarsest coherent configuration of the current structure $A_i := G_i \uplus H_i$ in the cloud. Since $M$ is normalised, $R(C(A_i)) \subseteq V(G_i)^2 \cup V(H_i)^2$. Thus, we can inductively apply \autoref{lem:PC_constructionOfNewAxiomSystem} to derive in $\EPC_3$ polynomial axiom systems $\Piso(G_i,H_i)$ for every $i \in [t]$. 
		The induction requires that for every vertex-colour (i.e.\ diagonal relation), the colour classes always have the same size in $G_i$ and $H_i$ (this is a prerequisite of \autoref{lem:PC_constructionOfNewAxiomSystem}). This is satisfied because $D(G_i) = D(H_i)$, for $i < t$.
		Since $D(G_t) \neq D(H_t)$, it follows from \autoref{lem:sketchesDistinctSpoilerWins} and \autoref{lem:PC_groheberkholzlemma} that the $1$-polynomial is derivable from  $\Piso(G_t,H_t)$ in the degree-$3$ monomial calculus, so in total, it is derivable from $\Piso(G,H)$ in $\EPC_3$. In order to apply \autoref{lem:PC_groheberkholzlemma} to $G_t$ and $H_t$, we need to argue that the vertex-colour-classes are of equal size in both graphs, even though $D(G_t) \neq D(H_t)$. If step $t$ is $\pairOp(R)$, then we introduce equally many new pair-vertices in both graphs because $D(G_{t-1}) = D(H_{t-1})$, so the numbers of $R$-pairs are equal. If step $t$ is $\sccOp(R)$, we also produce the same number of new vertices. Namely, the number of $R$-SCCs in $C(G_{t-1})$ and $C(H_{t-1})$ is equal, because by \autoref{lem:PC_SCCequalSize}, all $R$-SCCs have equal size, and by \autoref{lem:PC_SCCsSingleColoured}, vertices in $R$-SCCs receive the same diagonal colour $P$ distinct from all diagonal colours outside SCCs (and $|P(C(G_{t-1}))| = |P(C(H_{t-1}))|$).\\   
		
		Finally, we bound the number of extension variables used in the derivation of $\Piso(G_t,H_t)$: As stated in \autoref{lem:PC_constructionOfNewAxiomSystem}, for every $i$, the derivation of $\Piso(G_i,H_i)$ from $\Piso(G_{i-1}, H_{i-1})$ uses at most $|V(G_i)|^2$ many new extension variables.
		Therefore, the total number of extension variables that are used in the derivation of $\Piso(G_t,H_t)$ is at most $\sum_{i \in [t]} |V(G_i)|^2$. Since $t \leq p(|G| + |H|)$ and $|V(G_i)| \leq p(|G| + |H|)$, for every $i \in [t]$, this sum is at most $p(|G| + |H|)^3$. Similarly we can bound the size and bit-complexity of the derivation: Each time we invoke \autoref{lem:PC_constructionOfNewAxiomSystem}, we only incur a polynomial cost in size, and this happens polynomially many times. The occurring coefficients can be encoded with polynomially many bits as the lemma asserts. Also, \autoref{lem:PC_constructionOfNewAxiomSystem} uses only extension axioms of the required form.
		Finally, we bound the number of extension variables used in the derivation of $\Piso(G_t,H_t)$: As stated in \autoref{lem:PC_constructionOfNewAxiomSystem}, for every $i$, the derivation of $\Piso(G_i,H_i)$ from $\Piso(G_{i-1}, H_{i-1})$ uses at most $|V(G_i)|^2$ many new extension variables.
		Therefore, the total number of extension variables that are used in the derivation of $\Piso(G_t,H_t)$ is at most $\sum_{i \in [t]} |V(G_i)|^2$. Since $t \leq p(|G| + |H|)$ and $|V(G_i)| \leq p(|G| + |H|)$, for every $i \in [t]$, this sum is at most $p(|G| + |H|)^3$. Similarly we can bound the size and bit-complexity of the derivation: Each time we invoke \autoref{lem:PC_constructionOfNewAxiomSystem}, we only incur a polynomial cost in size, and this happens polynomially many times. The occurring coefficients can be encoded with polynomially many bits as the lemma asserts. Also, \autoref{lem:PC_constructionOfNewAxiomSystem} uses only extension axioms of the required form.
	\end{proof}
	\textit{Proof of \autoref{lem:PC_constructionOfNewAxiomSystem}:}\\
	First, we have to explain how the variables of the new system $\Piso(G',H')$ are encoded as polynomials in the old variables. 
	Recall that the variable set of $\Piso(G,H)$ is 
	\[
	\Vv(G,H) :=   \{ X_{vw} \mid v \in V(G), w \in V(H), v \text{ and } w \text{ have the same vertex-colour}  \}.
	\]
	The intended meaning of $X_{vw}$ is ``$v$ is mapped to $w$''. The graphs $G',H'$ contain new vertices, which either represent contracted $R$-SCCs or pairs of colour $R$.
	The set of vertex-pairs $(v,w)$ for which we need new variables is $\Pairs :=  (V(G') \setminus V(G)) \times (V(H') \setminus V(H))$.\\
	We would like to map each $(v,w) \in \Pairs$ to a polynomial $f(vw)$ such that we can represent variables $X_{vw}$ for $(v,w) \in \Pairs$ as extension variables $X_{f(vw)}$, which we can introduce with the extension axiom $X_{f(vw)} = f(vw)$. If $\op = \pairOp$ and $v$ is a new pair-vertex, then we let $\pair(v)$ be the vertex-pair that $v$ corresponds to. If $\op = \sccOp$ and $v$ is a new SCC-vertex, then we let $\scc(v)$ denote the set of vertices in the SCC that is contracted into $v$. We define $f$ as the following injective mapping $f : \Pairs \to \bbQ[\Vv(G,H)]$. 
	\begin{equation*}
		f(vw) := \begin{cases}
			X_{v_1w_1}X_{v_2w_2} & \text{, if } \op = \pairOp \text{ and }\\  & (v_1,v_2) = \pair(v), (w_1,w_2) = \pair(w). \\
			\frac{1}{|\scc(v)|} \cdot  \sum\limits_{(v',w')\in \scc(v) \times \scc(w)} X_{v'w'}& \text{, if } \op=\sccOp.
		\end{cases}
	\end{equation*} 
	Note that $f(vw)$ is indeed always a polynomial in variables $\Vv(G,H)$: To see this, we have to check that in the pair-case, the vertex-colours of $v_1,w_1$ and of $v_2, w_2$, respectively, are equal, and in the SCC-case, the vertex-colours of all elements of $\scc(v)$ and $\scc(w)$ are equal. In the pair-case, this follows from \autoref{cor:PC_endpointsColour}, and in the SCC-case from \autoref{cor:PC_SCCsSingleVertexColoured}.\\
	
	If $v$ and $w$ are newly introduced pair-vertices with $\pair(v) = (v_1,v_2)$ and $\pair(w) = (w_1,w_2)$, then the variable $X_{vw}$ will be the extension variable for the monomial $X_{v_1w_1}X_{v_2w_2}$. This makes sense because if $X_{vw}$ is set to $1$, then the bijection encoded by the assignment maps $v$ to $w$; but then it also has to map $v_1$ to $w_1$ and $v_2$ to $w_2$. Similarly, if $v$ and $w$ are new SCC-vertices, then any bijection that takes $v$ to $w$ must also map the elements of $\scc(v)$ to the elements of $\scc(w)$ in any possible way. This is reflected in our representation of $X_{vw}$ as the ``average'' over all possible mappings from $\scc(v)$ to $\scc(w)$.\\
	
	Here are the new polynomial axioms that we have to derive in order to go from $\Piso(G,H)$ to $\Piso(G',H')$:
	
	\begin{align}
		\label{liftedIsoAxiom1}
		&\sum_{v \in V(G') \setminus V(G)} X_{f(vw)} - 1 & &\text{ for all } w \in V(H') \setminus V(H).
	\end{align}
	\begin{align}
		\label{liftedIsoAxiom2}
		&\sum_{w \in V(H') \setminus V(H)} X_{f(vw)} - 1 & &\text{ for all } v \in V(G') \setminus V(G).
	\end{align}
	\begin{align}
		\label{liftedIsoAxiom3}
		&X_{f(vw)}X_{v'w'} & &\text{ for all } v,v' \in V(G'), w,w' \in V(H')\\
		& & &\text{ such that } (v,w) \in \Pairs  \notag\\
		& & &\text{ and } v', w' \in V(G) \cup V(H), v' \sim w' \notag\\ 
		& & &\text{ and } \{(v,w),(v',w')\} \text{ is not} \notag\\
		& & & \text{ a local isomorphism.} \notag
	\end{align}
	\begin{align}
		\label{liftedIsoAxiom4}
		&X_{f(vw)}X_{f(v'w')} & &\text{ for all } v,v' \in V(G'), w,w' \in V(H')\\
		& & &\text{ such that } (v,w) \in \Pairs \notag\\
		& & &\text{ and } (v',w') \in \Pairs \notag\\ 
		& & &\text{ and } \{(v,w),(v',w')\} \text{ is not} \notag\\
		& & & \text{ a local isomorphism.} \notag
	\end{align}
	The relation $\sim$ is the same-colour-relation, as in \autoref{def:PC_isoSystem}. Note that Axioms \eqref{liftedIsoAxiom1} and \eqref{liftedIsoAxiom2} only sum over vertices of the same colour as $w$ and $v$, respectively (as Axioms \eqref{isoAxiom1} and \eqref{isoAxiom2} do), because $V(G') \setminus V(G)$ and $V(H') \setminus V(H)$ are the sets of newly added vertices. These vertices receive a new colour distinct from all other vertex colours in $G$ and $H$ (see definition of the DWL-operations in Section \ref{sec:PC_DWLdefinition}). 
	The next step is to verify that $\Piso(G',H')$ is indeed derivable from $\Piso(G,H)$.\\
	\\
	\textbf{Derivation of Axioms \eqref{liftedIsoAxiom1} and \eqref{liftedIsoAxiom2}:}\\
	Fix $w \in V(H') \setminus V(H)$. We show how to derive Axiom \eqref{liftedIsoAxiom1} for $w$.
	Two cases have to be distinguished, namely whether $\op = \pairOp$ or $\op= \sccOp$.\\
	\textbf{\textit{Case 1:} $\op = \pairOp$:}
	Let $(w_1,w_2) = \pair(w)$. Using the multiplication rule and linear combinations, we derive from Axiom \eqref{isoAxiom1} for $w_2$ in $\Piso(G,H)$: 
	\begin{align*}
		\left(\sum\limits_{\stackrel{v_1 \in V(G),}{v_1 \sim w_1}} X_{v_1w_1}\right) &\cdot \left(\sum\limits_{\stackrel{v_2 \in V(G),}{v_2 \sim w_2}} X_{v_2w_2} - 1 \right)\\
		\\
		&=\sum\limits_{\stackrel{v_1,v_2 \in V(G)}{\stackrel{v_1 \sim w_1,}{v_2 \sim w_2}}} X_{v_1w_1}X_{v_2w_2} - \sum\limits_{\stackrel{v_1 \in V(G),}{v_1 \sim w_1}} X_{v_1w_1}
	\end{align*}
	Recall from the statement of \autoref{lem:PC_constructionOfNewAxiomSystem} that $R \in \sigma$ is the colour such that $\op(R)$ is executed to obtain $G' \uplus H'$ from $G \uplus H$. Further, let $C = C(G \uplus H)$.\\
	Since $(w_1,w_2) \in R(C)$, we can use \autoref{lem:PC_equivWLPebble} and \autoref{lem:PC_groheberkholzlemma} to derive from $\Piso(G,H)$ all monomials $X_{v_1w_1}X_{v_2w_2}$ where $(v_1,v_2) \notin R(C)$. Hence, we may cancel these monomials from the above sum with the linear combination rule. This yields:
	\[
	\sum\limits_{\stackrel{v_1,v_2 \in V(G)}{(v_1,v_2) \in R(C) }} X_{v_1w_1}X_{v_2w_2} - \sum\limits_{\stackrel{v_1 \in V(G)}{v_1 \sim w_1}} X_{v_1w_1}.
	\]
	Here, we used that for all pairs $(v_1,v_2) \in V(G)^2 \cap R(C(A))$, it holds that $v_1 \sim w_1$ and $v_2 \sim w_2$. This follows from \autoref{cor:PC_endpointsColour} and the fact that vertex-colours are represented by diagonal relations. Now we are almost done: We add Axiom \eqref{isoAxiom1} for $w_1$ to the above expression and replace each remaining monomial $X_{v_1w_1}X_{v_2w_2}$ with the new extension variable $X_{f(vw)}$, where $v \in V(G')$ is the respective new pair-vertex with $\pair(v) = (v_1,v_2)$. One can see that
	\[
	V(G') \setminus V(G) = \{  v \in V(G') \mid \pair(v) \in R(C) \cap V(G)^2   \}.
	\]
	Thus, we have indeed derived Axiom \eqref{liftedIsoAxiom1} for $w$.\\
	Similarly, we get Axiom \eqref{liftedIsoAxiom2} for a vertex $v \in V(G') \setminus V(G)$ if we perform the same derivations from the Axioms \eqref{isoAxiom2} instead of \eqref{isoAxiom1}.\\
	\\
	\textbf{\textit{Case 2:} $\op = \sccOp$:}
	We derive Axiom \eqref{liftedIsoAxiom1} for a fixed vertex $w \in V(H') \setminus V(H)$. Now $w$ is a vertex that represents a contracted $R$-SCC $\scc(w) \subseteq V(H)$.
	
	For every vertex $w' \in \scc(w)$, we have Axiom \eqref{isoAxiom1} for $w'$ in $\Piso(G,H)$:
	\[
	\sum_{\stackrel{v' \in V(G),}{ v' \sim w'}} X_{v'w'} - 1.
	\]
	Now from this, we may cancel all $X_{v'w'}$ where $(v',v')$ and $(w',w')$ are in distinct diagonal relations in $C$. This is done again by deriving the respective variables $X_{v'w'}$ with \autoref{lem:PC_equivWLPebble} and \autoref{lem:PC_groheberkholzlemma}. After that step, we have for each $w' \in \scc(w)$:
	\[
	\sum_{\stackrel{v' \in V(G),}{ \text{there is } v \in V(G') \setminus V(G) \text{ with } v' \in \scc(v) }} X_{v'w'} - 1.   \tag{$\star$}
	\]
	This holds because the vertex $v' \in V(G)$ has the same diagonal colour as $w' \in \scc(w)$ in the coherent configuration $C$ if and only if it is also contained in some $R$-SCC (\autoref{lem:PC_SCCsSingleColoured}).\\
	
	Next, we use the variable introduction rule and introduce the variables $X_{f(vw)}$ for every $v \in V(G') \setminus V(G)$. That means, we obtain the following polynomials:
	\begin{align*}
		&\frac{1}{|\scc(v)|}\sum\limits_{(v',w')\in \scc(v) \times \scc(w)} X_{v' w'} - X_{f(vw)} & &\text{ for each } v \in V(G') \setminus V(G).
	\end{align*}
	Now take the sum of all polynomials $(\star)$ for all $w' \in \scc(w)$, multiplied by $\frac{1}{|\scc(w)|}$. From this, subtract the above polynomials for all $v \in V(G') \setminus V(G)$. This yields Axiom \eqref{liftedIsoAxiom1} for the vertex $w$ because we have $\frac{1}{|\scc(v)|} = \frac{1}{|\scc(w)|}$ for all $v \in V(G') \setminus V(G)$, since all $R$-SCCs have equal size (\autoref{lem:PC_SCCequalSize}).
	In a similar way we can derive Axiom \eqref{liftedIsoAxiom2} for an SCC-vertex $v \in V(G') \setminus V(G)$.\\
	\\
\textbf{Derivation of Axioms \eqref{liftedIsoAxiom3}:}\\
Let $v,v' \in V(G'), w,w' \in V(H')$ such that $(v,w) \in \Pairs$ and $v',w' \in V(G) \cup V(H)$, and $v' \sim w'$. Furthermore, assume that $\{(v,w),(v',w')\}$ is not a local isomorphism. Since $v$ and $w$ are newly introduced vertices and $v',w'$ are old ones, it holds $v' \neq v$ and $w' \neq w$. Thus, if $\{(v,w),(v',w')\}$ is not a local isomorphism, there must be a relation symbol $E \in \tau$ such that $(v',v) \in E(G)$ and $(w',w) \notin E(H)$, or $(v,v') \in E(G)$ and $(w,w') \notin E(H)$, or vice versa. Again, we have to distinguish two cases:\\
\\
\textbf{\textit{Case 1:} $\op = \pairOp $:}\\
In this case, $v$ and $w$ are new pair-vertices representing pairs $(v_1,v_2)$ and $(w_1,w_2)$, respectively. Therefore, the only non-diagonal relations in which they occur are $E_{\text{left}}$ and $E_{\text{right}}$. Suppose $(v',v) \in E_{\text{left}}(G')$ and $(w',w) \notin E_{\text{left}}(H')$. That means $v' = v_1$ and $w' \neq w_1$. 
We take the extension axiom for $X_{f(vw)}$ and multiply it by $X_{v'w'}$ to obtain $X_{v'w'}(X_{v_1w_1} X_{v_2w_2} - X_{f(vw)})$.
Since $v' = v_1$ and $w' \neq w_1$, the monomial $X_{v'w'}X_{v_1w_1}$ represents a pebble position that is not a local isomorphism and is therefore an axiom in $\Piso(G,H)$. We can thus derive $X_{v'w'}X_{v_1w_1} X_{v_2w_2}$ and cancel it from the polynomial above. Then we multiply by $(-1)$ and are left with Axiom \eqref{liftedIsoAxiom3}, as desired.
Similarly, we can derive the axiom in the case that $(v,v') \in  E_{\text{right}}(G')$ and $(w,w') \notin E_{\text{right}}(H')$. The symmetric cases in which $(w,w')$ or $(w',w)$ is in the respective relation, and $(v,v')$ or $(v',v)$ is not, are analogous.\\ 
\\
\textbf{\textit{Case 2:} $\op = \sccOp$:}
In this case, $v$ and $w$ are contracted $R$-SCCs of $G$ and $H$. Let $E \in \tau$ be a relation symbol such that $(v',v) \in E(G')$ and  $(w',w) \notin E(H')$. Then by definition of $E(G' \uplus H')$ (see Section \ref{sec:PC_DWLdefinition}), there exists a $v_1 \in \scc(v)$ such that $(v',v_1) \in E(G)$, and there is no $w_1 \in \scc(w)$ such that $(w',w_1) \in E(H)$.
In order to derive $X_{f(vw)}X_{v'w'}$, we multiply the extension axiom
\[
\frac{1}{|\scc(v)|}\sum\limits_{(v'',w'')\in \scc(v) \times \scc(w)} X_{v'' w''} - X_{f(vw)}
\]
with $X_{v'w'}$. From the resulting sum, we can cancel all monomials of the form $X_{v'w'}X_{v''w''}$, for all $v'' \in \scc(v), w'' \in \scc(w)$, because $(v', v'')$ and $(w',w'')$ have distinct colours (using again \autoref{lem:PC_equivWLPebble} and \autoref{lem:PC_groheberkholzlemma}). The colours are distinct because there exists an $E$-edge from $v'$ into $\scc(v)$, but none from $w'$ into $\scc(w)$ (see \autoref{lem:PC_coloursBetweenSCCsDeterminedByEdgesVersion1}). After cancelling these monomials, we are left with $X_{f(vw)}X_{v'w'}$.
Again, the symmetric cases work analogously.\\
\\
\textbf{Derivation of Axioms \eqref{liftedIsoAxiom4}:}\\ 
Let $v,v' \in V(G'), w,w' \in V(H')$ such that $(v,w) \in \Pairs$ and  $(v',w') \in \Pairs$. Furthermore, assume that $\{(v,w),(v',w')\}$ is not a local isomorphism. Again, we distinguish between the two operation types:\\
\\
\textbf{\textit{Case 1:} $\op = \pairOp $:}\\
If all four vertices $v,v',w,w'$ are newly created pair-vertices, then $(v,v')$ and $(w,w')$ are not in any relation. Therefore, the only way how $\{(v,w),(v',w')\}$ can fail to be a local isomorphism is if $v = v'$ and $w \neq w'$ (or vice versa).\\
So let $v = v'$ and $\pair(v) = \pair(v') = (v_1,v_2)$. Further, let $\pair(w) = (w_1,w_2)$ and $\pair(w') = (w_1',w_2')$, where $\pair(w) \neq \pair(w')$. Suppose that $w_1 \neq w_1'$ (if $w_2 \neq w_2'$, the derivation is analogous). We multiply the extension axiom for $X_{f(vw)}$ with $X_{v_1w_1'}$ and obtain:
\[
X_{v_1w_1'}(X_{v_1w_1}X_{v_2w_2} - X_{f(vw)}).
\]
Since $w_1' \neq w_1$, the monomial $X_{v_1w_1'}X_{v_1w_1}$ encodes a pebble position which is not a local isomorphism and therefore, it is in $\Piso(G,H)$. Thus, we can derive $X_{v_1w_1'}X_{v_1w_1}X_{v_2w_2}$ and cancel it from the above polynomial, yielding $-X_{v_1w_1'}X_{f(vw)}$. Now multiply this by $X_{v_2w_2'}$ and add the result to the lifted extension axiom $X_{f(vw)}(X_{v_1w_1'}X_{v_2w_2'} -  X_{f(v'w')})$ (recall that $\pair(v') = (v_1,v_2)$). The result, multiplied by $(-1)$, is Axiom \eqref{liftedIsoAxiom4}, namely $X_{f(vw)}X_{f(v'w')}$.\\
Again, the symmetric cases are analogous.\\
\\
\textbf{\textit{Case 2:} $\op = \sccOp $:}\\
In this case, two subcases must be considered because there are two ways in which $\{(v,w),(v',w')\}$ can fail to be a local isomorphism.\\
\\
\textbf{\textit{Case 2.1: Mismatch of equality types.}}\\
Like in the previous case, let $v = v'$ and $w \neq w'$. 
We multiply the extension axiom for $X_{f(vw)}$ with a weighted sum of variables (using the  multiplication and the linear combination rule) to obtain:
\[
\left( \frac{1}{|\scc(v')|}\sum\limits_{v_1' \in \scc(v')}\sum\limits_{w_2 \in \scc(w')} X_{v_1'w_2} \right) \cdot \left(\frac{1}{|\scc(v)|}\sum\limits_{(v_1,w_1) \in \scc(v) \times \scc(w)}  X_{v_1  w_1} - X_{f(vw)}\right).
\]
Because $\scc(w) \cap \scc(w') = \emptyset$, and $\scc(v') = \scc(v)$, $w_1$ and $w_2$ in the above sum are always in distinct SCCs, while $v_1$ and $v_1'$ are in the same SCC. \autoref{lem:PC_differentColours} states that the colours of pairs in the same SCC are distinct from colours of pairs which do not lie in the same SCC. 
Hence, all monomials of the form $X_{v_1'w_2}X_{v_1 w_1}$ are derivable from $\Piso(G,H)$ using \autoref{lem:PC_equivWLPebble} and \autoref{lem:PC_groheberkholzlemma}. Cancelling these monomials from the above sum yields:
\[
\left( \frac{1}{|\scc(v')|}\sum\limits_{v_1' \in \scc(v')} \sum\limits_{w_2 \in \scc(w')} X_{v_1'w_2} \right) \cdot \left(- X_{f(vw)}\right).
\]
With the help of the extension axiom for $X_{f(v'w')}$, we can replace the sum in this expression by $X_{f(v'w')}$ and are done. Again, symmetric cases work analogously.\\
\\
\textbf{\textit{Case 2.2: Mismatch of relations.}}
In this case, the reason why $\{(v,w),(v',w')\}$ is not a local isomorphism is that there is a relation $E \in \tau$ such that $(v',v) \in E(G')$ and $(w',w) \notin E(H')$ (again, we skip the symmetric cases because they are analogous). Then by definition of $E(G' \uplus H')$, there exist $v_1' \in \scc(v')$ and $v_1 \in \scc(v)$ such that $(v_1',v_1) \in E(G)$, but for every pair $(w_1',w_1) \in \scc(w') \times \scc(w)$, it holds $(w_1',w_1) \notin E(H)$.\\
\\
We take the extension axiom for $X_{f(vw)}$ and multiply it with $X_{v_2w_2}$, for all $v_2 \in \scc(v'), w_2 \in \scc(w')$. This yields polynomials of the form (where we now write $v''$ for the vertices in $\scc(v)$ to avoid confusion with the vertex $v'$):
\[
\frac{1}{|\scc(v)|}\sum\limits_{(v'',w'') \in \scc(v) \times \scc(w)} X_{v''w''}X_{v_2w_2} - X_{f(vw)}X_{v_2w_2}. 
\]
We obtain such a polynomial for every $v_2 \in \scc(v'), w_2 \in \scc(w')$.\\
By \autoref{lem:PC_coloursBetweenSCCsDeterminedByEdges}, the pairs $(v'',v_2)$ and $(w'',w_2)$ have distinct colours in the coarsest coherent configuration $C$, for every $v'' \in \scc(v), w'' \in \scc(w)$, because there is an $E$-edge between $\scc(v')$ and $\scc(v)$, but none between $\scc(w')$ and $\scc(w)$. Therefore, each monomial $X_{v''w''}X_{v_2w_2}$ is derivable from $\Piso(G,H)$ and can be cancelled from the above sums.\\
So in total, we can derive: 
\[
- X_{f(vw)}X_{v_2w_2} \text{ , for all } v_2 \in \scc(v'), w_2 \in \scc(w').
\]
We use these monomials to cancel all the summands in the product of the extension axiom for $X_{f(v'w')}$ with the variable $X_{f(vw)}$, which is the following expression:
\[
\frac{1}{|\scc(v')|}\sum\limits_{(v_2,w_2) \in \scc(v') \times \scc(w')} X_{f(vw)}X_{f(v_2w_2)} - X_{f(vw)}X_{f(v'w')}.
\]
Cancelling out the summands as described yields the desired Axiom \eqref{liftedIsoAxiom4}: $X_{f(vw)}X_{f(v'w')}$.\\
	
In total, we can derive $\Piso(G',H')$ from $\Piso(G,H)$. The number of new variables is clearly bounded by $|V(G')|^2$. It is also not difficult to see that only polynomially many monomials occur in the derivation, and the binary encoding of the coefficients occurring in them has complexity at most $\Oo(\log |V(G)|)$. The used extension axioms are all for polynomials that are averaged sums or degree-2 monomials, as mentioned in \autoref{lem:PC_constructionOfNewAxiomSystem}.
The derivations obtained with \autoref{lem:PC_groheberkholzlemma} also have polynomial complexity because they can be carried out in $\MC_3$.

	\section{Discussion and future work}
	We have shown that the degree-$3$ extended polynomial calculus can simulate the pair- and contract-operations of Deep Weisfeiler Leman in the sense that the axiom system $\Piso(G',H')$ is derivable from $\Piso(G,H)$ if there is a sequence of DWL-operations that transforms $G \uplus H$ into $G' \uplus H'$. Together with the simulation of $k$-dimensional Weisfeiler Leman in the degree-$k$ monomial calculus given in \cite{berkholz}, this shows that $\EPC_3$ can distinguish two graphs $G$ and $H$ if they can be distinguished in DWL, and the $\EPC_3$-refutation has the same complexity as the DWL-algorithm. Since DWL-algorithms and CPT-programs mutually simulate each other, this result upper-bounds the graph distinguishing power of CPT by that of $\EPC_3$.\\
	
	This raises the question whether a super-polynomial lower bound for graph isomorphism can be established for $\EPC_3$, preferably for graph classes such as unordered CFI-graphs or multipedes, whose isomorphism problem reduces to a linear equation system and is thus in \ptime. If such a lower bound is found, then by \autoref{thm:PC_mainLowerBound}, we would also have that $\CPT \neq \ptime$. This would be a huge step forward in understanding the limitations of symmetry-invariant computation and thus in the quest for a logic for \ptime.\\
	
	Unfortunately, we do not know how strong the system $\EPC_3$ is, and in particular, if the degree-restriction is a true limitation. It may even be the case that $\EPC_3$ polynomially simulates the \emph{unbounded-degree} extended polynomial calculus. Then it would be as strong as extended Frege because in $\EPC$, the extension variables can encode arbitrary polynomials and thereby arbitrary Boolean circuits. This would make it less useful for proving lower bounds against $\CPT$, as extended Frege lower bounds seem to be out of reach at the moment.\\
	However, \autoref{thm:PC_mainInformal} also asserts that the simulation of $\CPT$ is possible using only extension axioms of a limited form, namely for degree-2 monomials and averaged sums. In this restricted version of $\EPC_3$, the obvious representation of Boolean circuits as polynomials is no longer possible: The Boolean functions $X \land Y, X \lor Y$, and $\neg X$ can naturally be represented as the polynomials $X \cdot Y, X + Y - X \cdot Y$, and $1-X$. When we represent Boolean circuits using extension variables, then each extension variable corresponds to a gate in the circuit. If the only allowed extension axioms are $\frac{}{X_f - f}$ for $f = X \cdot Y$ or $f = \frac{1}{n} \sum_{i=1}^{n^2} X_i$, then the only gates that we can naturally express are AND-gates (with extension axioms of the first type). Neither NOT-gates nor OR-gates can be simulated (directly) by such extension axioms because this requires sums which are not of the form $\frac{1}{n} \sum X_i$. In particular, these extension axioms cannot be applied to polynomials where variables occur with a negative coefficient. Hence, the corresponding circuits are in some sense \emph{monotone}. This is of course no formal proof that $\EPC_3$ with restricted extension axioms is strictly weaker than extended Frege but at least it rules out the natural simulation of Boolean circuits in $\EPC_3$. In total, the success chances of our suggested approach for $\CPT$ lower bounds via proof complexity depend highly on the true power of $\EPC_3$ (with restricted extension axioms), and its relation to unrestricted $\EPC$. Investigating this remains a problem for future work.\\
	
	\paragraph*{Symmetry-invariance of the refutations}
	Actually, our \autoref{thm:PC_mainInformal} could be strengthened more: A simulation of $\CPT$ in $\EPC_3$ is even possible in a certain \emph{symmetry-invariant} fragment of $\EPC_3$. However, it seems tricky to give a precise definition of ``\emph{symmetric} $\EPC_3$'' that is both natural and fits the kind of symmetry we encounter in our $\CPT$-simulation. 
	A neat way to put it would be to say that the set of extension axioms used in a derivation has to be closed under symmetries. With the right definition of ``symmetries'', this is indeed true for the refutation constructed in \autoref{lem:PC_constructionOfNewAxiomSystem}. Namely, whenever an extension variable $X_{f(vw)}$ is introduced, where $v$ and $w$ are new pair- or SCC-vertices, then we introduce it for \emph{all} $(v,w) \in V(G') \times V(H')$ that are new. The corresponding polynomials $f(vw)$ consist of variables that refer to the vertices in the respective pairs or SCCs of $v$ and $w$. The automorphisms of the graphs $G$ and $H$ preserve the colours of all vertex-pairs in the coarsest coherent configuration. Therefore, the set of extension axioms that we introduce in each step of the refutation is closed under the automorphisms of $G$ and $H$. The action of these automorphism groups on the set of variables of $\Piso(G,H)$ is the natural one, i.e.\ if $\pi$ is an automorphism of $G$ and $\sigma$ an automorphism of $H$, then they take $X_{vw}$ to $X_{\pi(v)\sigma(w)}$. This extends naturally to the extension axioms, so for example, if $v$ and $w$ are pair-vertices with $\pair(v) = (v_1,v_2), \pair(w) = (w_1,w_2)$, then the extension axiom $X_{f(vw)} - X_{v_1w_1}X_{v_2w_2}$ is mapped to $X_{f(v'w')} - X_{\pi(v_1)\sigma(w_1)}X_{\pi(v_2)\sigma(w_2)}$, where $v',w'$ are the newly introduced pair-vertices for $(\pi(v_1),\pi(v_2))$ and $(\sigma(w_1),\sigma(w_2))$ (such $v',w'$ must exist because DWL is isomorphism-invariant and introduces new vertices for all pairs with the same colour). So in this sense, the extension axioms used in \autoref{lem:PC_constructionOfNewAxiomSystem} are closed under all automorphism-pairs $(\pi, \sigma) \in \Aut(G) \times \Aut(H)$.\\
	Unfortunately, this does not lead to a \emph{general} definition of symmetric $\EPC_3$ because it depends on the automorphisms of $G$ and $H$, the graphs which are implicitly encoded in $\Piso(G,H)$. When $\EPC_3$ is applied to other polynomial axiom systems, then there might be no graphs ``in the background''. So for a general set of input polynomials $\Pp$, it would be natural to require that the set of extension axioms in a refutation be closed under the automorphisms of $\Pp$ -- these are the permutations of the variables that extend to permutations of the polynomials in $\Pp$. However, this would no longer fit to our derivation from \autoref{lem:PC_constructionOfNewAxiomSystem}: The system $\Piso(G,H)$ in general has more automorphisms than $\Aut(G) \times \Aut(H)$. Namely, $\Piso(G,H)$ contains no information about where the edges and non-edges in $G$ and $H$ actually are; it just relates pairs $(v,v') \in V(G)^2$ with pairs $(w,w') \in V(H)^2$ where $(v,v')$ is an edge and $(w,w')$ is not, or vice versa (Axiom \eqref{isoAxiom3}). Therefore, an automorphism of $\Piso(G,H)$ may swap all edges with non-edges, as long as it does so in both $G$ and $H$ (such examples can be constructed). But the automorphisms of the graphs must preserve edges and non-edges, so such an automorphism of $\Piso(G,H)$ does not correspond to one from $\Aut(G) \times \Aut(H)$. Our constructed refutation is only symmetric with respect to the latter. Thus, our simulation of $\CPT$ in $\EPC_3$ is possible in a way that respects \emph{specific} symmetries of $\Piso(G,H)$, but we do not know if this kind of symmetry-invariance can be formulated independently of the graph isomorphism problem as a general restriction to the proof system $\EPC_3$. Perhaps future research will reveal a more generic way to define symmetrized versions of known proof systems. This could be of independent interest because it might be possible to prove lower bounds for symmetric versions of proof systems for which non-symmetric lower bounds seem to be out of reach.\\
	
	Finally, another question that we have not answered is whether the converse to \autoref{thm:PC_mainInformal} also holds: Is there an algorithm that can find $\EPC_3$-refutations (for graph isomorphism) and can be implemented in $\CPT$? Since $\CPT$ is symmetry-invariant and $\EPC_3$ is not, this seems unlikely. Furthermore, such a proof search algorithm would probably have to be non-deterministic. Thus, the only way to get an exact match in expressive power between the logic and the proof system might be by restricting $\EPC_3$ to a symmetry-invariant fragment and extending $\CPT$ with some kind of non-determinism.
	
	\newpage
	
	\bibliographystyle{plainurl}
	
	\bibliography{references.bib}
		\newpage

		\section{Appendix}
		\label{sec:appendix}

		Here is a full proof of \autoref{lem:PC_equivWLPebble}, which adapts Theorem 2.2 in \cite{sandraSiglog} to the setting where we consider the disjoint union $G \uplus H$ rather than the two graphs separately. Essentially it works as expected with Spoiler's strategy being determined by the refinements made in the iterations of the $2$-\WL-algorithm. Additionally, we have to combine this with a technical insight from \cite{dwl} for handling disjoint unions of connected binary structures.
		\restatePebbleLemma*
		\begin{proof}
			We show the statement by induction on the number of iterations that $2$-dimensional Weisfeiler Leman needs to distinguish $(v,v')$ and $(w,w')$ in the structure $A$. Let us make precise how the $2$-WL algorithm computes $C(A)$ by iteratively refining colourings of $V(A)^2$. In the initial colouring $C_0$, there is only one diagonal colour $R_{\diag}$ with $R_{\diag}(C_0) := V(A)^2$. One colour $R_{\text{cross}}$ is reserved for all crossing pairs, i.e.\ $R_{\text{cross}}(C_0) = (V(G) \times V(H)) \cup (V(H) \times V(G)) $. The remaining pairs are coloured according to their atomic types, so there is one colour for each atomic type of pairs in $V(G)^2 \cup V(H)^2$ that is realised in $A$. The atomic type of a pair $(v,w)$ is the set of relations $R \in \tau$ such that $(v,w) \in R(A)$.
			Note that this initial colouring is not necessarily a coherent configuration: It satisfies all properties from \autoref{def:PC_coherentConfig} except the last one about intersection numbers. In fact, this is the case for all colourings that are computed throughout the iteration, except for the final one, which is stable and a coarsest coherent configuration of $A$. To argue why this resulting configuration is indeed equivalent to $C(A)$, it is important that $A$ is the disjoint union of two connected structures, and therefore, by Lemma 8 in \cite{dwl}, its crossing colours are distinct from its non-crossing colours. Therefore, the choice of our initial colouring $C_0$ will not lead to a stable colouring that is different from $C(A)$.\\
			
			The colouring $C_{i+1}$ is defined from the $\sigma_i$-colouring $C_i$ as follows: Each colour class $R(C)$ is split along the intersection numbers of its pairs with other colour classes. That means $R(C)$ is split into the coarsest possible partition $\{ P_1,...,P_m \}$ such that for each $P_i$ it holds: For all pairs $(u,v) \in P_i$, and all $S_1, S_2 \in \sigma_i$, the number of all $x \in V(A)$ such that $(u,x) \in S_1(C_i)$ and $(x,v) \in S_2(C_i)$ is the same (i.e.\ independent of the chosen pair in $P_i$).\\
			Refining every colour in $\sigma_i$ in this way yields the colouring $C_{i+1}$. We simply enumerate the colours in $\sigma_{i+1}$ and call them $R_1,R_2,...$ and so on, because we do not care about their actual names. This refinement process stops when the colouring is stable and cannot be refined further -- the resulting colouring is equivalent to $C(A)$, the canonical coarsest coherent configuration, as it is the coarsest possible colouring that also satisfies the last condition of \autoref{def:PC_coherentConfig}.\\
			
			For any relation $R \subseteq V(A)^2$, let $\dom(R) := \{x \in V(A) \mid \text{there exists } y \text{ such that }(x,y) \in R \}$. We show the following four statements via induction on the number $i$ of iterations of the refinement procedure:
			\begin{enumerate}[(a)]
				\item For every colour $R \in \sigma_i$, either $R(C_i) \subseteq (V(G) \times V(H)) \cup (V(H) \times V(G))$ or $R(C_i)$ is disjoint from $(V(G) \times V(H)) \cup (V(H) \times V(G))$. In the former case we say that $R$ is \emph{crossing}.
				\item For every colour $R \in \sigma_i$, there exist diagonal colours $D_1,D_2 \in \sigma_i$ such that for every pair $(u,v) \in R(C_i)$, it holds $(u,u) \in D_1(C_i)$ and $(v,v) \in D_2(C_i)$.
				\item Let $v,v' \in V(G), w,w' \in V(H)$ and let $D_1,D_2 \in \sigma_i$ be diagonal colours such that $(v,v), (w,w) \in D_1(C_i)$ and $(v',v'), (w',w') \in D_2(C_i)$. Then $(v,w)$ and $(v',w')$ have the same (crossing) colour.
				%Assume that there exists a bijection $g_{i-1}: V(G) \lra V(H)$ that preserves the $\sigma_{i-1}$-colours of all pairs in $V(G)^2$. Then it holds: For every \emph{crossing} colour $R \in \sigma_i$, there exist diagonal colours $D_1,D_2 \in \sigma_i$ such that $R(C_i) = \dom (D_1(C_i)) \times \dom (D_2(C_i))$. (If such a bijection does not exist, then Spoiler wins the game on $G$ and $H$ from any position by statement (d)).
				\item Let $(v,v') \in V(G)^2, (w,w') \in V(H)^2$ such that $(v,v')$ and $(w,w')$ do not have the same colour in $C_i$. Then Spoiler has a winning strategy for the bijective 3-pebble game on $G$ and $H$ with initial position $\{(v,w),(v',w')\}$.
			\end{enumerate}	 
		
			\textit{Proof:} For $i=0$, (a), (b) and (c) are clear by definition of $C_0$, and (d) is also clear since in the initial colouring, distinct colours mean distinct atomic types. In that case, $\{(v,w),(v',w')\}$ is not a local isomorphism and Spoiler wins immediately.\\
			
			Now consider iteration $i+1$. Statement (a) follows from the inductive hypothesis because the colouring is refined in every step and thus, each crossing colour is always partitioned into crossing colours again, and the same holds for non-crossing colours.\\
			
			Next, we show statement (b). Fix a colour $R \in \sigma_i$ and diagonal colours $D_1,D_2 \in \sigma_i$ such that for every pair $(u,v) \in R(C_i)$, it holds $(u,u) \in D_1(C_i)$ and $(v,v) \in D_2(C_i)$. We have to show: If any of the diagonal colours $D_1(C_i), D_2(C_i)$ are split, then the colour $R$ is split in such a way that statement (b) still holds after iteration $i+1$.
			Assume w.l.o.g.\ that $D_1(C_i)$ is split: Let $(u,u),(u',u') \in D_1(C_i)$ and let $v,v'$ be such that $(u,v), (u',v') \in R(C_i)$. Further, let $S_1, S_2 \in \sigma_i$ such that
			\begin{align*}
				|X | &:= |\{ x \in V(A) \mid (u,x) \in S_1(C_i) \text{ and } (x,u) \in S_2(C_i)  \}|\\
				& \neq |\{ x \in V(A) \mid (u',x) \in S_1(C_i) \text{ and } (x,u') \in S_2(C_i) \}| =: | X' | .
			\end{align*}
			Note that we have $S_2 = S_1^{-1}$.
			Now consider any two pairs $(u,v), (u',v') \in R(C_i)$. Partition $X$ according to the colours of its elements paired with $v$, i.e.\ for any colour $T \in \sigma_i$, let $X_T := \{ x \in X \mid (x,v) \in T \}$. Then the non-empty $X_T$ form a partition of $X$. Similarly, define $X'_T := \{ x \in X' \mid (x,v') \in T \}$. Since $|X| \neq |X'|$, there must exist a colour $T \in \sigma$ such that $|X_T| \neq |X'_T|$. Then for this colour, we have 
			\begin{align*}
				&|\{ x \in V(A) \mid (u,x) \in S_1(C_i) \text{ and } (x,v) \in T(C_i)  \}| = |X_T| \\
				\neq |X'_T| = &|\{ x \in V(A) \mid (u',x) \in S_1(C_i) \text{ and } (x,v') \in T(C_i) \}|.
			\end{align*}
			Thus, the pairs $(u,v), (u',v')$ are in distinct colours after iteration $i+1$, as witnessed by the intersection numbers with the colours $S$ and $T$. Hence, the invariant (b) still holds.\\
			
			We know that statement (c) holds after iteration $i$. In order to show that it still holds after iteration $i+1$, we need to prove that whenever a crossing colour is refined, then at least one of its endpoint-colours is also refined: Fix a crossing colour $R \in \sigma_i$ and two pairs $(u_1,u_2), (u_1',u_2') \in R(C_i)$. By statement (b), we know that there are diagonal colours $D_1, D_2$ such that $(u_1,u_1), (u_1',u_1') \in D_1(C_i)$ and $(u_2,u_2), (u_2',u_2') \in D_2(C_i)$. Now suppose that in iteration $i+1$, the pairs $(u_1,u_2)$ and $(u_1',u_2')$ are separated. Our goal is to show that also $u_1$ and $u_1'$ or $u_2$ and $u_2'$ get distinct diagonal colours because the only way how (c) can fail to be true is if $(u_1,u_2)$ and $(u_1',u_2')$ get distinct colours but their respective first and second entries keep the same diagonal colour as before. 
			So let $S_1, S_2 \in \sigma_i$ be colours that witness the separation of $(u_1,u_2)$ and $(u_1',u_2')$:  	
			\begin{align*}
				|X | &:= |\{ x \in V(A) \mid (u_1,x) \in S_1(C_i) \text{ and } (x,u_2) \in S_2(C_i)  \}|\\
				& \neq |\{ x \in V(A) \mid (u'_1,x) \in S_1(C_i) \text{ and } (x,u'_2) \in S_2(C_i) \}| =: | X' | .
			\end{align*}
			Assume w.l.o.g.\ that $|X| > 0$. Exactly one of the colours $S_1, S_2$ is crossing, and the other is non-crossing. Assume w.l.o.g.\ that $S_1$ is non-crossing and $S_2$ is crossing. Then $X,X' \subseteq V(G).$\\
			\textbf{Claim}: $X = \{ x \in V(A) \mid (u_1,x) \in S_1(C_i) \text{ and } (u_1,x) \in S_1^{-1}(C_i)  \}$.\\
			\textit{Proof of claim:} The inclusion $\subseteq$ is clear. For the inclusion $\supseteq$, we have to show that for every $x \in V(A)$ with $(u_1,x) \in S_1(C_i)$ it holds $(x,u_2) \in S_2(C_i)$. This is true because: The diagonal colour of $(x,x)$ is the same as that of every vertex in $X$, according to statement (b) with respect to $S_1$. Then statement (c) from the induction hypothesis implies that $(x,u_2) \in S_2(C_i)$. This proves the claim.\\
			\\
			Similarly, we can prove $X' = \{ x \in V(A) \mid (u'_1,x) \in S_1(C_i) \text{ and } (u'_1,x) \in S_1^{-1}(C_i)  \}$.\\
			Hence we have
			\begin{align*}
				& |\{ x \in V(A) \mid (u_1,x) \in S_1(C_i) \text{ and } (x,u_1) \in S_1^{-1}(C_i)  \}|\\
				\neq &|\{ x \in V(A) \mid (u'_1,x) \in S_1(C_i) \text{ and } (x,u'_1) \in S_1^{-1}(C_i) \}|.
			\end{align*}
			Therefore, $(u_1,u_1), (u'_1,u'_1) \in D_1(C_i)$ will get distinct diagonal colours after iteration $i+1$, as witnessed by the intersection numbers with $S_1$ and $S_1^{-1}$. If $S_2$ is non-crossing and $S_1$ is crossing, then it is the colour $D_2$ that is refined. This is what we wanted to show, so statement (c) is still true after iteration $i+1$.\\
			
			Finally, we can use this to prove statement (d). Assume that $(v,v') \in V(G)^2$ and $(w,w') \in V(H)^2$ have the same colour in $C_i$ and get distinct colours in $C_{i+1}$. Then there exist colours $R_1, R_2 \in \sigma_i$ such that 
			\begin{align*}
				| X_{vv'} | &:= |\{  x \in V(A) \mid (v,x) \in R_1(C_i) \text{ and } (x,v') \in R_2(C_i)   \}| \\
				&\neq |\{ x \in V(A) \mid (w,x) \in R_1(C_i)) \text{ and } (x,w') \in R_2(C_i)  \}| =: | X_{ww'} |.
			\end{align*}
			We distinguish two cases:\\
			\\
			\textit{Case 1:} $|X_{vv'} \cap V(G)| \neq |X_{ww'} \cap V(H)|$. In this case, Spoiler can play as follows from position $\{(v,w),(v',w')\}$: Let $f : V(G) \lra V(H)$ be the bijection chosen by Duplicator. If $|X_{vv'} \cap V(G)| > |X_{ww'} \cap V(H)|$, then Spoiler chooses some $x \in X_{vv'} \cap V(G)$ such that $f(x) \in V(H) \setminus X_{ww'}$, and if $|X_{vv'} \cap V(G)| < |X_{ww'} \cap V(H)|$, then he chooses $x \in V(G) \setminus X_{vv'}$ such that $f(x) \in X_{ww'} \cap V(H)$. In both cases, the resulting position $\{(v,w),(v',w'), (x,f(x)))\}$ is a winning position for Spoiler by the inductive hypothesis because either in $\{ (v,w), (x,f(x))  \}$ or in $\{ (v',w'), (x,f(x)) \}$, the pebble pairs have distinct colours in $C_i$.\\
			\\
			\textit{Case 2:} $|X_{vv'} \cap V(G)| = |X_{ww'} \cap V(H)|$. In this case, we have $|X_{vv'} \cap V(H)| \neq | X_{ww'} \cap V(G) |$. W.l.o.g.\ assume that $|X_{vv'} \cap V(H)| > 0$. It can be seen that $R_1$ and $R_2$ are crossing colours. By statement (b), there is a diagonal colour $D \in \sigma_i$ such that all vertices in $(X_{vv'} \cap V(H)) \cup (X_{ww'} \cap V(G))$ have the diagonal colour $D$, because these are the second entries of pairs in $R_1(C_i)$. Statement (c) says even more: For every vertex $x \in V(H)$ with $(x,x) \in D$, we have $(v,x) \in R_1(C_i)$ and $(x,v') \in R_2(C_i)$, and for every $x \in V(G)$ with $(x,x) \in D$, we have $(w,x) \in R_1(C_i)$ and $(x,w) \in R_2(C_i)$. Summarising these considerations, we get $X_{vv'} \cap V(H) = \dom (D) \cap V(H)$ and $X_{ww'} \cap V(G) = \dom(D) \cap V(G)$. Thus, we have  
			$|\dom (D) \cap V(H)| \neq | \dom(D) \cap V(G) |$. Then Spoiler wins the game on $G$ and $H$ from any starting position: He can enforce a position $\{  (y,z)   \}$ with $(y,y) \in D(C_i)$ and $(z,z) \notin D(C_i)$ (or vice versa). From there, he wins by the induction hypothesis.\\
			This finishes the inductive proof of (a) -- (d).
			\autoref{lem:PC_equivWLPebble} now follows from statement (d).
		\end{proof}

\end{document}